\documentclass[runningheads]{llncs}
\usepackage{graphicx}
\usepackage{marvosym}
\usepackage{amsmath}

\usepackage{mathtools}
\usepackage{dashrule}

\usepackage{graphicx,epsfig,color}

\usepackage{xypic}
\xyoption{rotate}
\input xy
\xyoption{all}
\usepackage{verbatim}
\usepackage{amssymb}
\usepackage{eucal}
\usepackage{enumerate}
\usepackage[shortlabels]{enumitem}

\usepackage{multicol}
\usepackage{manfnt}

\usepackage[latin1]{inputenc}

\usepackage{amsthm}
\usepackage{amsfonts}
\usepackage{mathrsfs}
\usepackage[all,2cell]{xy}
\CompileMatrices
\UseAllTwocells
\SilentMatrices
\usepackage{float}
\usepackage{fancybox}
\usepackage{hyperref}
\usepackage{wrapfig}
\usepackage{wasysym}

\usepackage{mathbbol}

\mathcode`:="003A  
\mathcode`;="003B  
\mathcode`?="003F  
\mathcode`|="026A  

\mathchardef\ls="213C    
\mathchardef\gr="213E    
\mathchardef\uparrow="0222  
\mathchardef\downarrow="0223  

\newcommand{\from}{\mathrel{:}}

\newcommand{\infcomp}{\mathrel{\uparrow}}
\newcommand{\ninfcomp}{\mathrel{\mathrlap{/}{\uparrow}}}

\newenvironment{Iff-RL}{\textbf{($\Rightarrow$)} }{\bigskip}
\newenvironment{Iff-LR}{\textbf{($\Leftarrow$)} }{}

\def \: {\colon}

\def \N {\mathbb{N}}
\def \Z {\mathbb{Z}}

\def \tns {\oplus}

\DeclareMathSymbol{\reversedExclMark}{\mathord}{operators}{"3C}

\newcommand\Seqcomp{\lower5pt\hbox{$\includegraphics[height=.8cm]{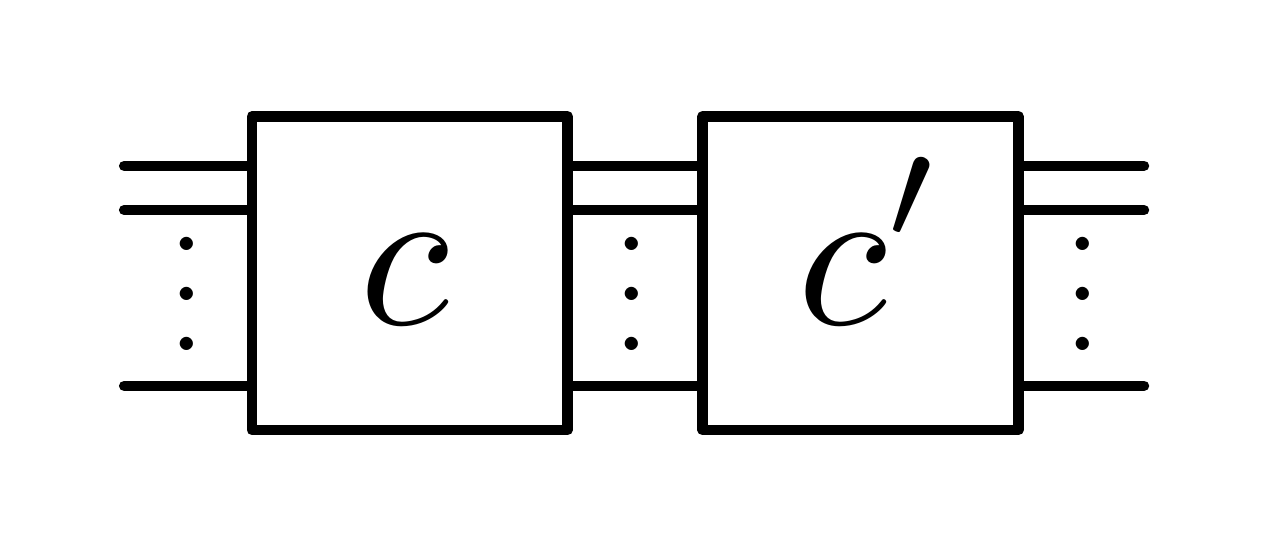}$}}
\newcommand\Parcomp{\lower5pt\hbox{$\includegraphics[height=.8cm]{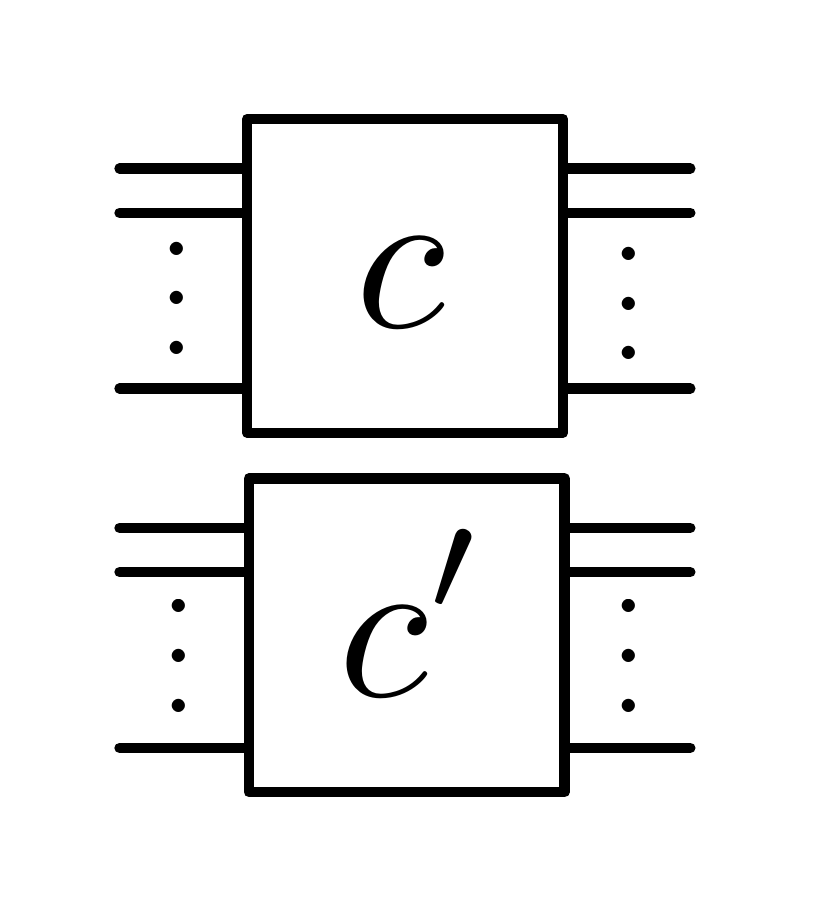}$}}

\newcommand\Gmult{\lower5pt\hbox{$\includegraphics[width=20pt]{graffles/Gmult.pdf}$}}
\newcommand\Gcomult{\lower5pt\hbox{$\includegraphics[width=20pt]{graffles/Gcomult.pdf}$}}
\newcommand\Gunit{\lower5pt\hbox{$\includegraphics[width=16pt]{graffles/Gunit.pdf}$}}
\newcommand\Gcounit{\lower5pt\hbox{$\includegraphics[width=16pt]{graffles/Gcounit.pdf}$}}
\newcommand\twoGcounit{\lower5pt\hbox{$\includegraphics[width=16pt]{graffles/twoGcounit.pdf}$}}

\newcommand\ASepUno{\lower3pt\hbox{$\includegraphics[width=28pt]{graffles/ASep1.pdf}$}}
\newcommand\ASepDue{\lower3pt\hbox{$\includegraphics[width=24pt]{graffles/ASep2.pdf}$}}

\newcommand\BSepOne{\lower3pt\hbox{\includegraphics[width=30pt]{graffles/BlackSep1.pdf}}}
\newcommand\BlackX{\lower5pt\hbox{\includegraphics[width=30pt]{graffles/BlackX.pdf}}}

\newcommand{\figref}[1]{Fig.~\ref{#1}}

\newcommand{\matrixNull}{\bullet}

\def \poly {\mathsf{k}[x]} 
\def \frpoly {\mathsf{k}(x)} 
\def \laur {\mathsf{k}((x))} 
\def \fps {\mathsf{k}[[x]]} 
\def \field {\mathsf{k}} 
\def \ratio {\mathsf{k}\langle x\rangle} 

\newcommand \LSum[1] {\sum_{i= #1}^{\infty}}

\newcommand \IH[1]{\mathbb{IH}_{\scriptscriptstyle #1}}

\def \AIH {\mathbb{aIH}} 

\def \poi {\,\ensuremath{;}\,}

\def \df {\ \ensuremath{:\!\!=}\ }

\newcommand\unitscalar{\lower4pt\hbox{$\includegraphics[width=22pt]{graffles/unitscalar.pdf}$}}
\newcommand\scalarminusone{\lower8pt\hbox{$\includegraphics[width=30pt]{graffles/scalarminusone.pdf}$}}
\newcommand\antipode{\lower3pt\hbox{$\includegraphics[width=22pt]{graffles/antipode.pdf}$}}
\newcommand\antipodeop{\lower3pt\hbox{$\includegraphics[width=22pt]{graffles/antipodeop.pdf}$}}
\newcommand\antipodesquare{\lower3pt\hbox{$\includegraphics[width=22pt]{graffles/antipodesquare.pdf}$}}
\newcommand\circuitAdots{\lower6pt\hbox{$\includegraphics[width=30pt]{graffles/circuitAdots.pdf}$}}

\newcommand\wcounitn{\lower5pt\hbox{$\includegraphics[width=25pt]{graffles/wcounitn.pdf}$}}
\newcommand\bcounitn{\lower5pt\hbox{$\includegraphics[width=25pt]{graffles/bcounitn.pdf}$}}
\newcommand\lccn{\lower5pt\hbox{$\includegraphics[width=25pt]{graffles/lccn.pdf}$}}
\newcommand\rccn{\lower5pt\hbox{$\includegraphics[width=25pt]{graffles/rccn.pdf}$}}
\newcommand\idncircuit{\lower5pt\hbox{$\includegraphics[width=25pt]{graffles/idncircuit.pdf}$}}
\newcommand\circuitrbcounits{\lower5pt\hbox{$\includegraphics[width=25pt]{graffles/circuitrbcounits.pdf}$}}
\newcommand\lccB{\lower5pt\hbox{$\includegraphics[width=25pt]{graffles/rccr.pdf}$}}
\newcommand\rccB{\lower5pt\hbox{$\includegraphics[width=25pt]{graffles/lccl.pdf}$}}
\newcommand\IdBcounitc{\lower5pt\hbox{$\includegraphics[width=20pt]{graffles/IdBcounit.pdf}$}}
\newcommand\BcounitId{\lower5pt\hbox{$\includegraphics[width=20pt]{graffles/BcounitId.pdf}$}}
\newcommand\symNetTwoOne{\lower7pt\hbox{$\includegraphics[width=25pt]{graffles/symNet21.pdf}$}}
\newcommand\nscalar{\!\!\lower5pt\hbox{$\includegraphics[width=35pt]{graffles/nscalar.pdf}$}\!\!}
\newcommand\Wmultstar{\!\lower5pt\hbox{$\includegraphics[width=20pt]{graffles/Wmultstar.pdf}$}\!}
\newcommand\scalarstar{\!\!\lower7pt\hbox{$\includegraphics[width=35pt]{graffles/scalarstar.pdf}$}\!\!}
\newcommand\twoBcounit{\!\!\lower8pt\hbox{$\includegraphics[width=20pt]{graffles/lunitsr.pdf}$}\!\!}

\newcommand\delay{\!\lower6pt\hbox{$\includegraphics[width=25pt]{graffles/delaycircuit.pdf}$}\!}
\newcommand\scalarp{\!\lower4pt\hbox{$\includegraphics[width=22pt]{graffles/scalarp.pdf}$}\!}
\newcommand\scalarpstar{\!\lower3pt\hbox{$\includegraphics[width=25pt]{graffles/scalarpstar.pdf}$}\!}
\newcommand\scalarpop{\!\lower4pt\hbox{$\includegraphics[width=22pt]{graffles/scalarpop.pdf}$}\!}
\newcommand\nscalarp{\!\lower3pt\hbox{$\includegraphics[width=30pt]{graffles/nscalarp.pdf}$}\!}
\newcommand\circuitfibr{\!\lower3pt\hbox{$\includegraphics[width=44pt]{graffles/circuitfibr.pdf}$}\!}
\newcommand\ncircuitX{\!\lower3pt\hbox{$\includegraphics[width=30pt]{graffles/ncircuitX.pdf}$}\!}
\newcommand\scalarpone{\!\lower3pt\hbox{$\includegraphics[width=22pt]{graffles/scalarpone.pdf}$}\!}
\newcommand\scalarptwoop{\!\lower3pt\hbox{$\includegraphics[width=22pt]{graffles/scalarptwoop.pdf}$}\!}

\tikzset{x=1em, y=1.5ex}
}

\newcommand{\Wunit}{
\tikzset{x=1em, y=2.1ex}
\begin{tikzpicture}[baseline=-.5ex, scale=.7]
	\begin{pgfonlayer}{nodelayer}
		\node [style=white] (0) at (0.25, 0) {};
		\node [style=none] (1) at (1.75, 0) {};
	\end{pgfonlayer}
	\begin{pgfonlayer}{edgelayer}
		\draw (0) to (1.center);
	\end{pgfonlayer}
\end{tikzpicture}
}
\tikzset{x=1em, y=1.5ex}
}
\newcommand{\Bcounit}{
\tikzset{x=1em, y=2.1ex}
\begin{tikzpicture}[baseline=-.5ex, scale=.7]
	\begin{pgfonlayer}{nodelayer}
		\node [style=black] (0) at (1.5, 0) {};
		\node [style=none] (1) at (0, 0) {};
	\end{pgfonlayer}
	\begin{pgfonlayer}{edgelayer}
		\draw (0) to (1.center);
	\end{pgfonlayer}
\end{tikzpicture}
}
\tikzset{x=1em, y=1.5ex}
}
 \newcommand{\Bcomult}{
\tikzset{x=1em, y=2.1ex}
\InputIfFileExists{./generators/copy.tikz}{}{\input{./tikz/./generators/copy.tikz}}
\tikzset{x=1em, y=1.5ex}
}
\newcommand{\Wmult}{
\tikzset{x=1em, y=2.1ex}
\InputIfFileExists{./generators/add.tikz}{}{\input{./tikz/./generators/add.tikz}}
\tikzset{x=1em, y=1.5ex}
}
\newcommand{\scalar}{
\tikzset{x=1em, y=2.1ex}
\begin{tikzpicture}
	\begin{pgfonlayer}{nodelayer}
		\node [style=reg] (0) at (-0.25, 0) {$k$};
		\node [style=none] (1) at (1, 0) {};
		\node [style=none] (2) at (-1.5, 0) {};
	\end{pgfonlayer}
	\begin{pgfonlayer}{edgelayer}
		\draw (2.center) to (0);
		\draw (0) to (1.center);
	\end{pgfonlayer}
\end{tikzpicture}
}
\tikzset{x=1em, y=1.5ex}
}
\newcommand{\circuitX}{
\tikzset{x=1em, y=2.1ex}
\begin{tikzpicture}
	\begin{pgfonlayer}{nodelayer}
		\node [style=reg] (0) at (-0.25, -0) {$x$};
		\node [style=none] (1) at (1, -0) {};
		\node [style=none] (2) at (-1.5, -0) {};
	\end{pgfonlayer}
	\begin{pgfonlayer}{edgelayer}
		\draw (2.center) to (0);
		\draw (0) to (1.center);
	\end{pgfonlayer}
\end{tikzpicture}}
\tikzset{x=1em, y=1.5ex}
}
\newcommand{\Bunit}{
\tikzset{x=1em, y=2.1ex}
\begin{tikzpicture}[baseline=-.5ex, scale=.7]
	\begin{pgfonlayer}{nodelayer}
		\node [style=black] (0) at (0.25, 0) {};
		\node [style=none] (1) at (1.75, 0) {};
	\end{pgfonlayer}
	\begin{pgfonlayer}{edgelayer}
		\draw (0) to (1.center);
	\end{pgfonlayer}
\end{tikzpicture}
}
\tikzset{x=1em, y=1.5ex}
}

\tikzset{x=1em, y=1.5ex}
}

\tikzset{x=1em, y=1.5ex}
}

\tikzset{x=1em, y=1.5ex}
}

\newcommand{\IdnetT}{
\tikzset{x=1em, y=2.1ex}
}
\tikzset{x=1em, y=1.5ex}
}
\newcommand{\symNetT}{
\tikzset{x=1em, y=2.1ex}
\InputIfFileExists{./generators/sym.tikz}{}{\input{./tikz/./generators/sym.tikz}}
\tikzset{x=1em, y=1.5ex}
}
\newcommand{\WunitT}{
\tikzset{x=1em, y=2.1ex}
}
\tikzset{x=1em, y=1.5ex}
}
\newcommand{\BcounitT}{
\tikzset{x=1em, y=2.1ex}
}
\tikzset{x=1em, y=1.5ex}
}
 \newcommand{\BcomultT}{
\tikzset{x=1em, y=2.1ex}
\InputIfFileExists{./generators/copy.tikz}{}{\input{./tikz/./generators/copy.tikz}}
\tikzset{x=1em, y=1.5ex}
}
\newcommand{\WmultT}{
\tikzset{x=1em, y=2.1ex}
\InputIfFileExists{./generators/add.tikz}{}{\input{./tikz/./generators/add.tikz}}
\tikzset{x=1em, y=1.5ex}
}
\newcommand{\scalarT}{
\tikzset{x=1em, y=2.1ex}
\begin{tikzpicture}
	\begin{pgfonlayer}{nodelayer}
		\node [style=reg] (0) at (-0.25, 0) {$k$};
		\node [style=none] (1) at (1, 0) {};
		\node [style=none] (2) at (-1.5, 0) {};
	\end{pgfonlayer}
	\begin{pgfonlayer}{edgelayer}
		\draw (2.center) to (0);
		\draw (0) to (1.center);
	\end{pgfonlayer}
\end{tikzpicture}
}
\tikzset{x=1em, y=1.5ex}
}
\newcommand{\circuitXT}{
\tikzset{x=1em, y=2.1ex}
\begin{tikzpicture}
	\begin{pgfonlayer}{nodelayer}
		\node [style=reg] (0) at (-0.25, -0) {$x$};
		\node [style=none] (1) at (1, -0) {};
		\node [style=none] (2) at (-1.5, -0) {};
	\end{pgfonlayer}
	\begin{pgfonlayer}{edgelayer}
		\draw (2.center) to (0);
		\draw (0) to (1.center);
	\end{pgfonlayer}
\end{tikzpicture}}
\tikzset{x=1em, y=1.5ex}
}
\newcommand{\BunitT}{
\tikzset{x=1em, y=2.1ex}
\begin{tikzpicture}[baseline=-.5ex, scale=.7]
	\begin{pgfonlayer}{nodelayer}
		\node [style=black] (0) at (0.25, 0) {};
		\node [style=none] (1) at (1.75, 0) {};
	\end{pgfonlayer}
	\begin{pgfonlayer}{edgelayer}
		\draw (0) to (1.center);
	\end{pgfonlayer}
\end{tikzpicture}
}
\tikzset{x=1em, y=1.5ex}
}
\newcommand{\WcounitT}{
\tikzset{x=1em, y=2.1ex}
}
\tikzset{x=1em, y=1.5ex}
}
\newcommand{\WcomultT}{
\tikzset{x=1em, y=2.1ex}
\InputIfFileExists{./generators/co-add.tikz}{}{\input{./tikz/./generators/co-add.tikz}}
\tikzset{x=1em, y=1.5ex}
}
\newcommand{\BmultT}{
\tikzset{x=1em, y=2.1ex}
\InputIfFileExists{./generators/co-copy.tikz}{}{\input{./tikz/./generators/co-copy.tikz}}
\tikzset{x=1em, y=1.5ex}
}
\newcommand{\scalaropT}{
\tikzset{x=1em, y=2.1ex}
\begin{tikzpicture}
	\begin{pgfonlayer}{nodelayer}
		\node [style=coreg] (0) at (-0.25, 0) {$k$};
		\node [style=none] (1) at (1, 0) {};
		\node [style=none] (2) at (-1.5, 0) {};
	\end{pgfonlayer}
	\begin{pgfonlayer}{edgelayer}
		\draw (2.center) to (0);
		\draw (0) to (1.center);
	\end{pgfonlayer}
\end{tikzpicture}
}
\tikzset{x=1em, y=1.5ex}
}
\newcommand{\circuitXopT}{
\tikzset{x=1em, y=2.1ex}
\begin{tikzpicture}
	\begin{pgfonlayer}{nodelayer}
		\node [style=coreg] (0) at (-0.25, 0) {$x$};
		\node [style=none] (1) at (1, 0) {};
		\node [style=none] (2) at (-1.5, 0) {};
	\end{pgfonlayer}
	\begin{pgfonlayer}{edgelayer}
		\draw (2.center) to (0);
		\draw (0) to (1.center);
	\end{pgfonlayer}
\end{tikzpicture}
}
\tikzset{x=1em, y=1.5ex}
}
\newcommand{\ZeronetT}{
\tikzset{x=1em, y=2.1ex}
\InputIfFileExists{./generators/empty-diag.tikz}{}{\input{./tikz/./generators/empty-diag.tikz}}
\tikzset{x=1em, y=1.5ex}
}

\newcommand{\circuitminusone}{\!\lower5pt\hbox{$\includegraphics[width=20pt,height=15pt]{graffles/circuitminusone.pdf}$}\!}
\newcommand{\circuitminusoneop}{\!\lower5pt\hbox{$\includegraphics[width=20pt,height=15pt]{graffles/circuitminusoneop.pdf}$}\!}

\newcommand{\typ}{\mathrel{:}}

\newcommand{\lbbd}{\mathopen{[\![}}
\newcommand{\rbbd}{\mathclose{]\!]}}
\newcommand{\dsem}[1]{\lbbd #1 \rbbd}
\newcommand{\dsemO}{\dsem{\cdot}}

\newcommand{\lbbo}{\mathopen{\langle}}
\newcommand{\rbbo}{\mathclose{\rangle}}
\newcommand{\osem}[1]{\lbbo #1 \rbbo}

\newcommand\idzcircuit{\lower5pt\hbox{$\includegraphics[width=20pt]{graffles/idzcircuit.pdf}$}}
\newcommand{\circuitXspan}{\!\lower4pt\hbox{$\includegraphics[width=40pt]{graffles/circuitXspan.pdf}$}\!}
\newcommand{\circuitXcospan}{\!\lower5pt\hbox{$\includegraphics[width=40pt]{graffles/circuitXcospan.pdf}$}\!}
\newcommand\zeroscalar{\lower3pt\hbox{$\includegraphics[width=20pt]{graffles/zeroscalar.pdf}$}}
\newcommand\zeroscalarr{\lower3pt\hbox{$\includegraphics[width=25pt]{graffles/zeroscalar2.pdf}$}}
\newcommand\rationalcircuit{\lower5pt\hbox{$\includegraphics[width=35pt]{graffles/rationalcircuit.pdf}$}}
\newcommand\Wccl{\lower5pt\hbox{$\includegraphics[width=22pt]{graffles/Wccl.pdf}$}}
\newcommand{\circuitkkop}{\!\lower4pt\hbox{$\includegraphics[width=32pt]{graffles/circuitkkop.pdf}$}\!}

\usepackage{centernot}

\newcommand\circuitUnoMinusX{\lower4.5pt\hbox{$\includegraphics[width=32pt]{graffles/circuit1-x.pdf}$}}
\newcommand\circuitUnoMinusXSquare{\lower5pt\hbox{$\includegraphics[width=35pt]{graffles/circuit1-xsquare.pdf}$}}

\usepackage{prooftree}

\newcommand{\ruleLabel}[1]{#1}
\newcommand{\labelSep}{\,}

\newcommand{\bnfEq}{\; ::= \;}
\newcommand{\bnfSep}{\;\; | \;\;}

\makeatletter
\def\moverlay{\mathpalette\mov@rlay}
\def\mov@rlay#1#2{\leavevmode\vtop{%
\baselineskip\z@skip \lineskiplimit-\maxdimen
\ialign{\hfil$#1##$\hfil\cr#2\crcr}}}
\makeatother

\newcommand{\fullcontext}[2]{\tcbox[size=small,on line, colback=black!10,  colframe=black!15, extrude right by=-1pt, extrude left by=-1pt, sharp corners=east, valign=center]{$#1\;\triangleright\;#2$}}

\newcommand{\context}[1]{\tcbox[size=small,on line, colback=black!10,  colframe=black!15, extrude right by=-2pt, extrude left by=-.5pt, sharp corners=east, valign=center]{$#1\;\triangleright$}}

\newcommand{\typeJudgment}[3]{{ {#2} \,\typ\, {#3}}}
\newcommand{\sort}[2]{\ensuremath{(#1,\,#2)}}

 \newcommand{\derivationRule}[3]{{\prooftree{ #1}\justifies{ #2}\using\ruleLabel{#3}\endprooftree}}
\newcommand{\reductionRule}[2]{{\prooftree{\scriptstyle #1}\justifies{\scriptstyle #2}\endprooftree}}

\newcommand\twarr[2]{%
\mathrel{\mathop{\moverlay{\scriptstyle\xrightarrow{\,#1\,}\cr{\lower.2em\hbox{$\scriptstyle{}_{#2}$}}}}}}
\newcommand\twarrw[2]{%
\mathrel{\mathop{\moverlay{\scriptstyle\Longrightarrow\cr{\lower-.6em\hbox{$\scriptstyle{}_{#1}$}}
\cr{\lower.3em\hbox{$\scriptstyle{}_{#2}$}}}}}}

\newcommand{\dtrans}[2]{\hbox{$\;\twarr{#1}{#2}\;$}}
\newcommand{\dtransw}[2]{\raise1pt\hbox{$\;\twarrw{#1}{#2}\;$}}

\def \CD {\mathsf{Circ}}
\newcommand{\FC}{\mathsf{C}\overrightarrow{\hspace{-.1cm}\scriptstyle\mathsf{irc}}}

\usepackage{bbm}

\newcommand{\eqAIH}{\stackrel{\tiny \AIH }{=}}

\def\SFGform{\mathsf{SF}}
\def\ASFG{\mathsf{ASF}}
\newcommand\Tr[1]{\mathsf{Tr}^{#1}} 

\newcommand{\Traj}{\mathsf{Traj}}

\usepackage{tikz}
\usetikzlibrary{arrows,backgrounds,shapes.geometric,shapes.misc,matrix,arrows.meta,decorations.markings}

\pgfdeclarelayer{background}
\pgfdeclarelayer{edgelayer}
\pgfdeclarelayer{nodelayer}
\pgfsetlayers{background,edgelayer,nodelayer,main}
\tikzset{baseline=-0.5ex}
\definecolor{light-gray}{gray}{.7}
\tikzstyle{none}=[inner sep=0pt]
\tikzstyle{plain}=[inner sep=0pt]
\tikzstyle{black}=[circle, draw=black, fill=black, inner sep=0pt, minimum size=3.5pt]
\tikzstyle{black-faded}=[circle, draw=light-gray, fill=light-gray, inner sep=0pt, minimum size=4pt]
\tikzstyle{white}=[circle, draw=black, fill=white, inner sep=0pt, minimum size=3.5pt]
\tikzstyle{white-faded}=[circle, draw=light-gray, fill=white, inner sep=0pt, minimum size=4.5pt]
\tikzstyle{delay}=[fill=black, regular polygon, regular polygon sides=3,rotate=-90, scale=.55]
\tikzstyle{delay-op}=[fill=black, regular polygon, regular polygon sides=3,rotate=90, scale=.55]
\tikzstyle{reg}=[draw, fill=white, rounded rectangle, rounded rectangle left arc=none, minimum height=1em, minimum width=1em, node font={\scriptsize}]
\tikzstyle{coreg}=[draw, fill=white, rounded rectangle, rounded rectangle right arc=none, minimum height=1em, minimum width=1em, node font={\scriptsize}]
\tikzstyle{rn}=[circle, draw=red, fill=red, inner sep=0pt, minimum size=4pt]
\tikzstyle{place}=[circle, draw=black, fill=white, inner sep=0pt, minimum size=8pt]

\newcommand{\ctikzfig}[1]{%
\begin{center}
  
\tikzset{x=1em, y=2.1ex}
\InputIfFileExists{#1.tikz}{}{\input{./tikz/#1.tikz}}
\tikzset{x=1em, y=1.5ex}

\end{center}}

\usepackage[most]{tcolorbox}
\usepackage{xstring}

\newcommand{\boxtikzfig}[2]{
\tikzset{x=1em, y=2.1ex}
\InputIfFileExists{#1.tikz}{}{\input{./tikz/#1.tikz}}
\tikzset{x=1em, y=1.5ex}
}

\newcommand{\myeq}[1]{\mathrel{\overset{\makebox[0pt]{\mbox{\normalfont\scriptsize\sffamily (#1)}}}{=}}}

\newcommand{\One}{
\tikzset{x=1em, y=2.1ex}
\begin{tikzpicture}[baseline=-.5ex]
	\begin{pgfonlayer}{nodelayer}
		\node [style=none] (0) at (-0.75, -0) {};
		\node [style=none] (1) at (0.25, -0) {};
		\node [style=none] (2) at (-0.75, 0.25) {};
		\node [style=none] (3) at (-0.75, -0.25) {};
	\end{pgfonlayer}
	\begin{pgfonlayer}{edgelayer}
		\draw (0.center) to (1.center);
		\draw (3.center) to (2.center);
	\end{pgfonlayer}
\end{tikzpicture}}
\tikzset{x=1em, y=1.5ex}
}
\newcommand{\coOne}{
\tikzset{x=1em, y=2.1ex}
\begin{tikzpicture}[baseline=-.5ex]
	\begin{pgfonlayer}{nodelayer}
		\node [style=none] (0) at (0.25, 0) {};
		\node [style=none] (1) at (-0.75, 0) {};
		\node [style=none] (2) at (0.25, 0.25) {};
		\node [style=none] (3) at (0.25, -0.25) {};
	\end{pgfonlayer}
	\begin{pgfonlayer}{edgelayer}
		\draw (0.center) to (1.center);
		\draw (3.center) to (2.center);
	\end{pgfonlayer}
\end{tikzpicture}
}
\tikzset{x=1em, y=1.5ex}
}
\def \ACD {\mathsf{ACirc}}

\newcommand{\Obseq}{\equiv}

\newcommand \ARel[1]{{\mathsf{ARel}}_{\scriptscriptstyle #1}}

\newcommand{\tcoscalar}[1]{
\begin{tikzpicture}[x=1em, y=2.1ex]
	\begin{pgfonlayer}{nodelayer}
		\node [style=coreg] (0) at (-0.25, 0) {$#1$};
		\node [style=none] (1) at (1.25, 0) {};
		\node [style=none] (2) at (-1.75, 0) {};
	\end{pgfonlayer}
	\begin{pgfonlayer}{edgelayer}
		\draw (2.center) to (0);
		\draw (0) to (1.center);
	\end{pgfonlayer}
\end{tikzpicture}}

\newcommand{\tscalar}[1]{
\begin{tikzpicture}[x=1em, y=2.1ex]
	\begin{pgfonlayer}{nodelayer}
		\node [style=reg] (0) at (-0.25, 0) {$#1$};
		\node [style=none] (1) at (1.25, 0) {};
		\node [style=none] (2) at (-1.75, 0) {};
	\end{pgfonlayer}
	\begin{pgfonlayer}{edgelayer}
		\draw (2.center) to (0);
		\draw (0) to (1.center);
	\end{pgfonlayer}
\end{tikzpicture}}

\newcommand{\stregister}[1]{
\begin{tikzpicture}[x=1em, y=2.1ex]
	\begin{pgfonlayer}{nodelayer}
		\node [style=reg] (0) at (-0.25, 0) {$x$};
		\node [style=none] (1) at (1.25, 0) {};
		\node [style=none] (2) at (-1.75, 0) {};
		\node [style=none] (3) at (-0.25, 1) {\scriptsize $#1$};
	\end{pgfonlayer}
	\begin{pgfonlayer}{edgelayer}
		\draw (2.center) to (0);
		\draw (0) to (1.center);
	\end{pgfonlayer}
\end{tikzpicture}
}
\newcommand{\stcoregister}[1]{
\begin{tikzpicture}[x=1em, y=2.1ex]
	\begin{pgfonlayer}{nodelayer}
		\node [style=coreg] (0) at (-0.25, 0) {$x$};
		\node [style=none] (1) at (1.25, 0) {};
		\node [style=none] (2) at (-1.75, 0) {};
		\node [style=none] (3) at (-0.25, 1) {\scriptsize $#1$};
	\end{pgfonlayer}
	\begin{pgfonlayer}{edgelayer}
		\draw (2.center) to (0);
		\draw (0) to (1.center);
	\end{pgfonlayer}
\end{tikzpicture}
}

\newcommand{\cospanregs}[2]{
\begin{tikzpicture}[x=1em, y=2.1ex]
	\begin{pgfonlayer}{nodelayer}
		\node [style=reg] (0) at (-0.25, 0) {$x$};
		\node [style=coreg] (1) at (1.25, 0) {$x$};
		\node [style=none] (2) at (-1.5, 0) {};
		\node [style=none] (3) at (-0.25, 1) {\scriptsize $#1$};
		\node [style=none] (4) at (2.5, 0) {};
		\node [style=none] (5) at (1.25, 1) {\scriptsize $#2$};
	\end{pgfonlayer}
	\begin{pgfonlayer}{edgelayer}
		\draw (2.center) to (0);
		\draw (0) to (1);
		\draw (1) to (4.center);
	\end{pgfonlayer}
\end{tikzpicture}
}

\newcommand{\spanregs}[1]{
\begin{tikzpicture}[x=1em, y=2.1ex]
	\begin{pgfonlayer}{nodelayer}
		\node [style=coreg] (0) at (-0.25, 0) {$x$};
		\node [style=reg] (1) at (1.25, 0) {$x$};
		\node [style=none] (2) at (-1.5, 0) {};
		\node [style=none] (3) at (-0.25, 1) {\scriptsize $#1$};
		\node [style=none] (4) at (2.5, 0) {};
		\node [style=none] (5) at (1.25, 1) {\scriptsize $#1$};
	\end{pgfonlayer}
	\begin{pgfonlayer}{edgelayer}
		\draw (2.center) to (0);
		\draw (0) to (1);
		\draw (1) to (4.center);
	\end{pgfonlayer}
\end{tikzpicture}
}
\newcommand{\onetwocoregs}[2]{
\begin{tikzpicture}[x=1em, y=2.1ex]
	\begin{pgfonlayer}{nodelayer}
		\node [style=none] (0) at (-1.75, 0) {};
		\node [style=coreg] (1) at (-0.5, 0) {$x$};
		\node [style=none] (2) at (-1.75, 0.25) {};
		\node [style=none] (3) at (-1.75, -0.25) {};
		\node [style=coreg] (4) at (1, 0) {$x$};
		\node [style=none] (5) at (2.25, 0) {};
		\node [style=none] (6) at (-0.5, 1) {\scriptsize $#1$};
		\node [style=none] (7) at (1, 1) {\scriptsize $#2$};
	\end{pgfonlayer}
	\begin{pgfonlayer}{edgelayer}
		\draw (0.center) to (1);
		\draw (3.center) to (2.center);
		\draw (4) to (1);
		\draw (4) to (5.center);
	\end{pgfonlayer}
\end{tikzpicture}}

\newcommand{\onecoreg}[1]{
\begin{tikzpicture}[x=1em, y=2.1ex]
	\begin{pgfonlayer}{nodelayer}
		\node [style=none] (0) at (-1.75, 0) {};
		\node [style=coreg] (1) at (-0.5, 0) {$x$};
		\node [style=none] (2) at (-1.75, 0.25) {};
		\node [style=none] (3) at (-1.75, -0.25) {};
		\node [style=none] (4) at (0.75, 0) {};
		\node [style=none] (6) at (-0.5, 1) {\scriptsize $#1$};
	\end{pgfonlayer}
	\begin{pgfonlayer}{edgelayer}
		\draw (0.center) to (1);
		\draw (3.center) to (2.center);
		\draw (4.center) to (1);
	\end{pgfonlayer}
\end{tikzpicture}}

\begin{document}
\title{Contextual Equivalence for Signal Flow Graphs}
%
%
\author{Filippo Bonchi\inst{1}
\and
Robin Piedeleu\inst{2}\thanks{Supported by EPSRC grant EP/R020604/1.}
\and
Pawe{\l} Soboci{\'n}ski\inst{3}\thanks{Supported
by the ESF funded Estonian IT Academy research measure
(project 2014-2020.4.05.19-0001)}
\and
Fabio Zanasi\inst{2}$^\star$
}
\authorrunning{F. Bonchi et al.}
%
\institute{Universit\`a di Pisa, Italy
\and
University College London, UK,
\email{\{r.piedeleu, f.zanasi\}@ucl.ac.uk}\and
Tallinn University of Technology, Estonia
}
\maketitle              
\begin{abstract}
We extend the signal flow calculus---a compositional account of the classical signal flow graph
model of computation---to encompass affine behaviour, and furnish it with a novel
operational semantics. The increased expressive power allows us to define
a canonical notion of contextual equivalence, which we show to coincide with denotational equality.
Finally, we characterise the realisable fragment of the calculus: those terms that express the computations
of (affine) signal flow graphs.
\keywords{signal flow graphs  \and affine relations \and full abstraction \and contextual equivalence \and string diagrams}
\end{abstract}

\section{Introduction}

Compositional accounts of models of computation often lead one to consider \emph{relational} models because a decomposition of an input-output system might consist of internal parts where flow and causality are not always easy to assign. These insights led Willems~\cite{Willems2007} to introduce a new current of control theory, called \emph{behavioural} control: roughly speaking, behaviours and observations are of prime concern, notions such as state, inputs or outputs are secondary. Independently, programming language theory converged on similar ideas, with \emph{contextual equivalence}~\cite{morris1969lambda,plotkin1975call} often considered as \emph{the} equivalence: programs are judged to be different if we can find some context in which one behaves differently from the other, and what is observed about ``behaviour'' is often something quite canonical and simple, such as termination. Hoare~\cite{Hoare1985} and Milner~\cite{Milner1980} discovered that these programming language theory innovations also bore fruit in the non-deterministic context of concurrency. 
Here again, research converged on studying simple and canonical contextual equivalences~\cite{Milner1992a,Honda1995}.

This paper brings together all of the above threads. The model of computation of interest for us is that of signal flow graphs~\cite{Shannon1942,mason1953feedback}, which are feedback systems well known in control theory~\cite{mason1953feedback} and widely used in the modelling of linear dynamical systems (in continuous time) and signal processing circuits (in discrete time). 
The \emph{signal flow calculus}~\cite{BonchiSZ17,Bonchi2015} is a syntactic presentation with an underlying compositional  denotational semantics in terms of linear relations. Armed with \emph{string diagrams}~\cite{Selinger2009} as a syntax, the tools and concepts of programming language theory and concurrency theory can be put to work and the calculus can be equipped with a structural operational semantics. However, while in previous work~\cite{Bonchi2015} a connection was made between operational equivalence (essentially trace equivalence) and denotational equality, the signal flow calculus was not quite expressive enough for contextual equivalence to be a useful notion.

The crucial step turns out to be moving from \emph{linear} relations to \emph{affine} relations, i.e. linear subspaces translated by a vector. In recent work~\cite{BonchiPSZ19}, we showed that they can be used to study important physical phenomena, such as current and voltage sources in electrical engineering, as well as fundamental synchronisation primitives in concurrency, such as mutual exclusion. Here we show that, in addition to yielding compelling mathematical domains, affinity proves to be the magic ingredient that ties the different components of the story of signal flow graphs together: it provides us with a canonical and simple notion of observation to use for the \emph{definition} of contextual equivalence, and gives us the expressive power to prove a bona fide full abstraction result that relates contextual equivalence with denotational equality.

To obtain the above result, we extend the signal flow calculus to handle affine behaviour. While the denotational semantics and axiomatic theory appeared in~\cite{BonchiPSZ19}, the operational account appears here for the first time and requires some technical innovations: instead of traces, we consider \emph{trajectories}, which are infinite traces that may start in the past. To record the time, states of our transition system have a runtime environment that keeps track of the global clock.

Because the affine signal flow calculus is oblivious to flow directionality, some terms exhibit pathological operational behaviour. We illustrate these phenomena with several examples. Nevertheless, for the linear sub-calculus, it is known~\cite{Bonchi2015} that every term is denotationally equal to an executable realisation: one that is in a form where a consistent flow can be identified, like the classical notion of signal flow graph. We show that the question has a more subtle answer in the affine extension: not all terms are realisable as (affine) signal flow graphs. However, we are able to characterise the class of diagrams for which this is true.

\vspace{-.1cm}

\paragraph{Related work.} Several authors studied signal flow graphs by exploiting concepts and techniques of programming language semantics, see e.g.~\cite{Prak2014,DBLP:conf/lics/Milius10,DBLP:journals/tcs/Rutten05,BaezErbele-CategoriesInControl}. The most relevant for this paper is \cite{BaezErbele-CategoriesInControl}, which, independently from~\cite{BonchiSZ17}, proposed the same syntax and axiomatisation for the ordinary signal flow calculus and shares with our contribution the same methodology: the use of
\emph{string diagrams} as a mathematical playground for the compositional study of different sorts of systems.
The idea is common to diverse, cross-disciplinary research programmes, including Categorical Quantum Mechanics~\cite{Abramsky2004,Coecke2008,Coecke2017}, Categorical Network Theory~\cite{Baez2014}, Monoidal Computer~\cite{Pavlovic13,Pavlovic14a} and the analysis of (a)synchronous circuits~\cite{Ghica2013,Ghica2016}.

\vspace{-.1cm}

\paragraph{Outline} In Section~\ref{sec:SFcalculus} we recall the affine signal flow calculus. Section~\ref{sec:opsem} introduces the operational semantics for the calculus. Section~\ref{sec:fullabs} defines contextual equivalence and proves full abstraction. Section~\ref{sec:sfg} introduces a well-behaved class of circuits, that denotes functional input-output systems, laying the groundwork for Section~\ref{sec:realisability}, in which the concept of realisability is introduced before a characterisation of which circuit diagrams are realisable. Missing proofs are presented in Appendix~\ref{app:proofs}.

\section{Background: the Affine Signal Flow Calculus}\label{sec:SFcalculus}

The \emph{Affine Signal Flow Calculus} extends the signal flow calculus \cite{Bonchi2015} with an extra generator $\One$ that allows to express affine relations. In this section, we first recall its syntax and denotational semantics from \cite{BonchiPSZ19} and then we highlight two key properties for proving full abstraction that are enabled by the affine extension. The operational semantics is delayed to the next section.

\begin{figure*}[ht]
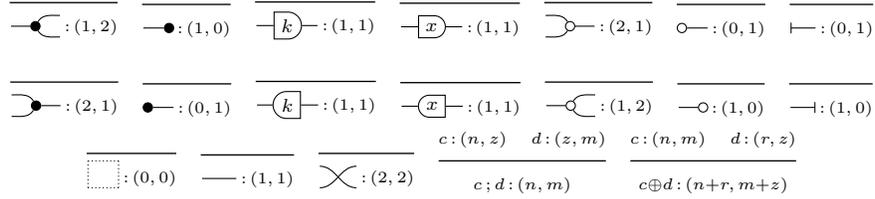

\[
\reductionRule{}{ \typeJudgment{}{\BcomultT}{\sort{1}{2}} }\quad
\reductionRule{}{ \typeJudgment{}{\BcounitT}{\sort{1}{0}} }\quad
\reductionRule{}{ \typeJudgment{}{\scalarT}{\sort{1}{1}} }\quad
\reductionRule{}{ \typeJudgment{}{\circuitXT}{\sort{1}{1}} }\quad
\reductionRule{}{ \typeJudgment{}{\WmultT}{\sort{2}{1}} }\quad
\reductionRule{}{ \typeJudgment{}{\WunitT}{\sort{0}{1}} }\quad
\reductionRule{}{ \typeJudgment{}{\One}{\sort{0}{1}} }
\]
\[
\reductionRule{}{ \typeJudgment{}{\BmultT}{\sort{2}{1}} }\quad
\reductionRule{}{ \typeJudgment{}{\BunitT}{\sort{0}{1}} } \quad
\reductionRule{}{ \typeJudgment{}{\scalaropT}{\sort{1}{1}} }\quad
\reductionRule{}{ \typeJudgment{}{\circuitXopT}{\sort{1}{1}} }\quad
\reductionRule{}{ \typeJudgment{}{\WcomultT}{\sort{1}{2}} }\quad
\reductionRule{}{ \typeJudgment{}{\WcounitT}{\sort{1}{0}} }\quad
\reductionRule{}{ \typeJudgment{}{\coOne}{\sort{1}{0}} }
\]
\[
\reductionRule{}{ \typeJudgment{}{\ZeronetT}{\sort{0}{0}} }\quad
\reductionRule{}{ \typeJudgment{}{\IdnetT}{\sort{1}{1}} }\quad
\reductionRule{}{ \typeJudgment{}{\symNetT}{\sort{2}{2}} }\quad
\reductionRule{ \typeJudgment{}{c}{\sort{n}{z}} \quad \typeJudgment{}{d}{\sort{z}{m}} }
{ \typeJudgment{}{c\poi d}{\sort{n}{m}} }\quad
\reductionRule{ \typeJudgment{}{c}{\sort{n}{m}} \quad \typeJudgment{}{d}{\sort{r}{z}} }
{ \typeJudgment{}{c \tns d}{\sort{n+r}{m+z}} }
\]
\caption{Sort inference rules.\label{fig:sortInferenceRules}}
\end{figure*}

\vspace{-.2cm}

\subsection{Syntax}\label{sec:syntax}

\vspace{-.2cm}

\begin{align}
c \bnfEq &  \BcounitT \bnfSep \BcomultT \bnfSep \scalarT \bnfSep \circuitXT   \bnfSep \WmultT \bnfSep \WunitT \bnfSep \One \bnfSep \label{eq:SFcalculusSyntax1} \\
& \,\BunitT \!\bnfSep \BmultT \bnfSep \scalaropT \bnfSep \!\circuitXopT \bnfSep  \!\WcomultT \bnfSep \WcounitT \bnfSep \coOne \bnfSep \label{eq:SFcalculusSyntax2}\\ 
& \,\ZeronetT \bnfSep \IdnetT \bnfSep \!\!\symNetT   \bnfSep c\tns c \bnfSep c \poi c \label{eq:SFcalculusSyntax3}
\end{align}
The syntax of the calculus, generated by the grammar above, is parametrised over a given field $\field$, with $k$ ranging over $\field$. We refer to the constants in rows~\eqref{eq:SFcalculusSyntax1}-\eqref{eq:SFcalculusSyntax2} as \emph{generators}. 
Terms are constructed from generators, $\ZeronetT$, $\IdnetT$, $\symNetT$, and the two binary operations in~\eqref{eq:SFcalculusSyntax3}.
We will only consider those terms that are \emph{sortable}, i.e. they can be associated with a pair  $\sort{n}{m}$, with $n,m\in \N$. Sortable terms are called \emph{circuits}: intuitively, a circuit with sort $\sort{n}{m}$ has $n$ ports on the left and $m$ on the right. The sorting discipline is given in Fig.~\ref{fig:sortInferenceRules}. We delay discussion of computational intuitions to Section~\ref{sec:opsem} but, for the time being, we observe that the generators of row \eqref{eq:SFcalculusSyntax2} are those of row \eqref{eq:SFcalculusSyntax1} ``reflected about the $y$-axis''.

\vspace{-.2cm}

\subsection{String Diagrams} 

\vspace{-.2cm}

It is convenient 
to consider circuits as the arrows of a 
symmetric monoidal category $\ACD$ (for Affine Circuits). Objects of $\ACD$ are natural numbers (thus $\ACD$ is a \emph{prop}~\cite{MacLane1965}) and morphisms $n \to m$ are the circuits of sort $\sort{n}{m}$, quotiented by the laws of symmetric monoidal categories \cite{mclane,Selinger2009}\footnote{This quotient is harmless: both the denotational semantics from \cite{BonchiPSZ19} and the operational semantics we introduce in this paper satisfy those axioms on the nose.}. The circuit grammar yields the symmetric monoidal structure of $\ACD$: sequential composition is given by $c \poi d$, the monoidal product is given by $c \tns d$, and identities and symmetries are built by pasting together $\IdnetT$ and $\symNetT$ in the obvious way. We will adopt the usual convention of writing morphisms of $\ACD$ as \emph{string diagrams}, meaning that  
$
c \poi c' \text{ is drawn }
\lower9pt\hbox{$\includegraphics[height=.8cm]{graffles/seqcomp.pdf}$}
\;
\text{ and } c \tns c' \text{ is drawn }
\lower15pt\hbox{$\includegraphics[height=1.2cm]{graffles/tensor.pdf}$}.
$
More succinctly, $\ACD$ is the free prop on generators \eqref{eq:SFcalculusSyntax1}-\eqref{eq:SFcalculusSyntax2}.
The free prop on \eqref{eq:SFcalculusSyntax1}-\eqref{eq:SFcalculusSyntax2} sans $\One$ and $\coOne$, hereafter called $\CD$, is the signal flow calculus from \cite{Bonchi2015}.

\begin{example}\label{ex:loop}
The diagram
{\lower12pt\hbox{$\includegraphics[height=1.2cm]{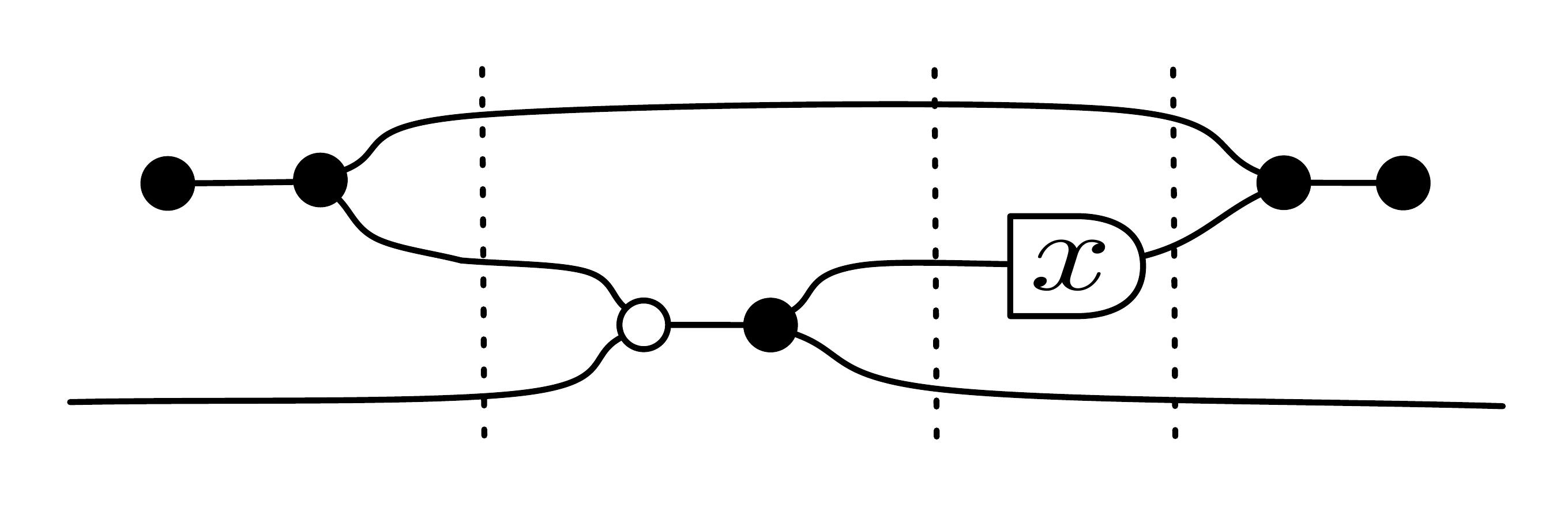}$}}
represents the circuit 
\[ ( (\BunitT\poi \BcomultT) \tns \IdnetT) \poi
(\IdnetT \tns (\WmultT \poi \BcomultT))
\poi (((\IdnetT \tns \circuitXT) \tns \IdnetT)
\poi ( (\BmultT \poi \BcounitT) \tns \IdnetT)).\] 
\end{example}

\subsection{Denotational Semantics and Axiomatisation}\label{sec:denotational} 

The semantics of circuits can be given denotationally by means of affine relations.

\begin{definition}
Let $\field$ be a field. An affine subspace of $\field^d$ is a subset $V\subseteq \field^d$ that is either empty or for which  there exists a vector $a\in\field^d$ and a linear subspace $L$ of $\field^d$ such that $V = \{a+v \mid v\in L\}$.
A \emph{$\field$-affine relation} of type $n \to m$ is an affine subspace of $\field^n\times\field^m$, considered as a $\field$-vector space.
\end{definition}
Note that every linear subspace is affine, taking $a$ above to be the zero vector. 
Affine relations can be organised into a prop:

\begin{definition}\label{DEF_SV} Let $\field$ be a field. Let $\ARel{\field}$ be the following prop:
 \begin{itemize}
\item arrows $n\to m$ are $\field$-affine relations.
 \item composition is relational: given
   $
   G=\{(u,v)\,|\,u\in\field^n,v\in\field^m\} \text{ and }
   H=\{(v,w)\,|\,v\in\field^m,w\in\field^l\}$,
    their composition is $G\poi H \df \{(u,w)\,|\,\exists v. (u,v)\in G \wedge (v,w)\in H\}$.
   \item monoidal product given by $G\tns H = \left\{
   \left(\left(\begin{array}{c}\!\!u\!\! \\ \!\!u'\!\!\end{array}\right),
   \left(\begin{array}{c}\!\!v\!\! \\ \!\!v'\!\!\end{array}\right)\right) \,|\, (u,v)\in G, (u',v')\in H\right\}$.
\end{itemize}
\end{definition}

In order to give semantics to $\ACD$, we use the prop of affine relations over the field $\frpoly$ of fractions of polynomials in $x$ with coefficients from $\field$. Elements $q\in \frpoly$ are a fractions $\frac{k_0+k_1\cdot x^1 + k_2 \cdot x^2 + \dots + k_n \cdot x^n}{l_0+l_1\cdot x^1 + l_2 \cdot x^2 + \dots + l_m \cdot l^m}$ for some $n,m\in \N$ and $k_i,l_i\in \field$. Sum, product, $0$ and $1$ in $\frpoly$ are defined as usual.

\begin{definition}\label{def:SemIB}
    The prop morphism $\dsemO \: \ACD \to \ARel{\frpoly}$ is inductively defined on circuits as follows. For the generators in \eqref{eq:SFcalculusSyntax1}
\[
\Bcomult \ \longmapsto \ \left\{ %
				\left( p ,\left(\begin{array}{c}
				 \! p\! \\
				 \! p\!
				\end{array}\right)\right) \mid p \in \frpoly \right\} 
\qquad 
\Wmult \ \longmapsto \   \left\{ \left( \left(%
				\begin{array}{c}
				  \!p\! \\
				  \!q\!
				\end{array}\right)\!\!, p+q\right) \mid p,q \in \frpoly \right\} 
\]
\[
\Bcounit \ \longmapsto \ \{ (p, \matrixNull ) \mid p\in\frpoly \}\qquad\quad
\Wunit  \ \longmapsto \  \{( \matrixNull ,0) \}\qquad\qquad\One \; \longmapsto \{ (\matrixNull, 1 )\}
\]
\vspace{-5pt}
\[
\tscalar{r}  \ \longmapsto \  \{ (p , p \cdot r) \mid p \in \frpoly \} \qquad
\circuitXT  \ \longmapsto \  \{ (p , p \cdot x) \mid p \in \frpoly \}\]
where $\matrixNull$ is the only element of $\frpoly^0$.
The semantics of components in \eqref{eq:SFcalculusSyntax2} is symmetric, e.g. $\Bunit$ is mapped to $\{ (p, \matrixNull ) \mid p\in\frpoly \}$. For \eqref{eq:SFcalculusSyntax3}
\begin{align*}
\IdnetT\!  \ \longmapsto \ \{(p,p) \mid p \in \frpoly \} 
\qquad
\symNetT  \ \longmapsto \  \left\{\left( \left(%
				\begin{array}{c}
				  \!p \!\\
				  \!q\!
				\end{array}\right), \left(%
				\begin{array}{c}
				 \!q\!\\
				\!p\!
				\end{array}\right)\right) \mid p, q \in \frpoly \right\} \\
\ZeronetT  \ \longmapsto \  \{ (\matrixNull , \matrixNull)\} 
\qquad
c_1\tns c_2\ \longmapsto\ \dsem{c_1} \tns\dsem{c_2}  \qquad 
c_1\poi c_2\ \longmapsto\ \dsem{c_1} \poi\dsem{c_2}
\end{align*}
\end{definition}
The reader can easily check that the pair of $1$-dimensional vectors $\left(1, \frac{1}{1-x}\right)\in \frpoly^1\times \frpoly^1$ belongs to the denotation of the circuit in Example \ref{ex:loop}.

The denotational semantics enjoys a sound and complete axiomatisation. The axioms involve only basic interactions between the generators~\eqref{eq:SFcalculusSyntax1}-\eqref{eq:SFcalculusSyntax2}. The resulting theory is that of \emph{Affine Interacting Hopf Algebras} ($\AIH$).
The generators in \eqref{eq:SFcalculusSyntax1} form a Hopf algebra, those in \eqref{eq:SFcalculusSyntax2} form another Hopf algebra, and the interaction of the two give rise to two Frobenius algebras. 
We recall the full set of equations in the Appendix~\ref{sec:axiomatisation}
We refer the reader to \cite{BonchiPSZ19} for the full set of equations and all further details.

\begin{proposition} For all $c,d$ in $\ACD$, $\dsem{c} = \dsem{d}$ if and only if $c \eqAIH d$.\end{proposition}

\subsection{Affine vs Linear Circuits}

It is important to highlight the differences between $\ACD$ and $\CD$. The latter is the purely linear fragment: circuit diagrams of $\CD$ denote exactly the \emph{linear} relations over $\frpoly$ \cite{Bonchi2014b}, while those of $\ACD$ denote the  \emph{affine} relations over $\frpoly$. 

The additional expressivity afforded by affine circuits is essential for our development. 
One crucial property is that every polynomial fraction can be expressed as an affine circuit of sort $\sort{0}{1}$. 
\begin{lemma}\label{prop:laurentCircuit}
For all $p \in \frpoly$, there is $c_{p}\in \ACD[0,1]$ with $\dsem{c_{p}}=\{(\matrixNull, p)\}$.
\end{lemma}
\begin{proof}
For each $p \in \frpoly$, let $P$ be the linear subspace generated by the pair of $1$-dimensional vectors $(1,p)$. By fullness of  the denotational semantics of $\CD$ \cite{Bonchi2014b}, there exists a circuit $c$ in $\CD$  such that $\dsem{c}=P$. Then, $\dsem{\One \poi c} = \{(\matrixNull, p)\}$. 
\end{proof}

The above observation yields the following:
\begin{proposition}\label{prop:denotation-context}
Let $(u,v)\in  \frpoly^n\times  \frpoly^m$. There exist circuits $c_u\in \ACD[0,n]$ and $c_v\in \ACD[m,0]$ such that $\dsem{c_u}=\{(\matrixNull, u)\}$ and $\dsem{c_v}=\{(v,\matrixNull)\}$.
\end{proposition}
\begin{proof}
Let $u={\tiny\left(%
				\begin{array}{c}
				  \!p_1\! \\
				  \!\vdots\! \\
				  \!p_n\!
				\end{array}\right)} \text{ and } v={\tiny\left(%
				\begin{array}{c}
				  \!q_1\! \\
				  \!\vdots\! \\
				  \!q_m\!
				\end{array}\right)}\text{.}$
By Lemma \ref{prop:laurentCircuit}, for each $p_i$, there exists a circuit $c_{p_i}$ such that $\dsem{c_{p_i}}=\{(\matrixNull, p_i)\}$. Let $c_u = c_{p_1} \tns \dots \tns c_{p_n}$. Then $\dsem{c_u}=\{(\matrixNull, u)\}$. For $c_v$, it is enough to see that Proposition \ref{prop:laurentCircuit} also holds with $0$ and $1$ switched, then use the argument above.
\end{proof}

Proposition~\ref{prop:denotation-context} 
 asserts that any behaviour $(u,v)$ occurring in the denotation of some circuit $c$, i.e., such that $(u,v)\in\dsem{c}$, can be expressed by a pair of circuits $(c_u,c_v)$. We will, in due course, think of such a pair as a \emph{context}, namely an environment with which a circuit can interact.
Observe that this is not possible with the linear fragment $\CD$, since the only singleton linear subspace is $0$.

Another difference between linear and affine concerns circuits of sort $\sort{0}{0}$. Indeed $\frpoly^0=\{\matrixNull\}$, and the only linear relation over $\frpoly^0\times \frpoly^0$ is the singleton $\{(\matrixNull, \matrixNull)\}$, which is $id_0$ in $\ARel{\frpoly}$. But there is another affine relation, namely the \emph{empty relation} $\emptyset \in \frpoly^0\times \frpoly^0$. This can be represented by $
\tikzset{x=1em, y=2.1ex}
\begin{tikzpicture}[baseline=-.5ex]
	\begin{pgfonlayer}{nodelayer}
		\node [style=none] (0) at (-0.75, -0) {};
		\node [style=white] (1) at (0.25, -0) {};
		\node [style=none] (2) at (-0.75, -0.25) {};
		\node [style=none] (3) at (-0.75, 0.25) {};
	\end{pgfonlayer}
	\begin{pgfonlayer}{edgelayer}
		\draw (0.center) to (1);
		\draw (3.center) to (2.center);
	\end{pgfonlayer}
\end{tikzpicture}}
\tikzset{x=1em, y=1.5ex}
$, for instance, since  $\dsem{
\tikzset{x=1em, y=2.1ex}
}
\tikzset{x=1em, y=1.5ex}
}\,=\{(\bullet,1)\}\poi\{(0,\bullet)\} = \emptyset$.

\begin{proposition}\label{prop:twopossibility}
Let $c\in \ACD[0,0]$. Then $\dsem{c}$ is either $id_0$ or $\emptyset$.
\end{proposition}

\section{Operational Semantics for Affine Circuits}\label{sec:opsem}

\begin{figure*}[ht]
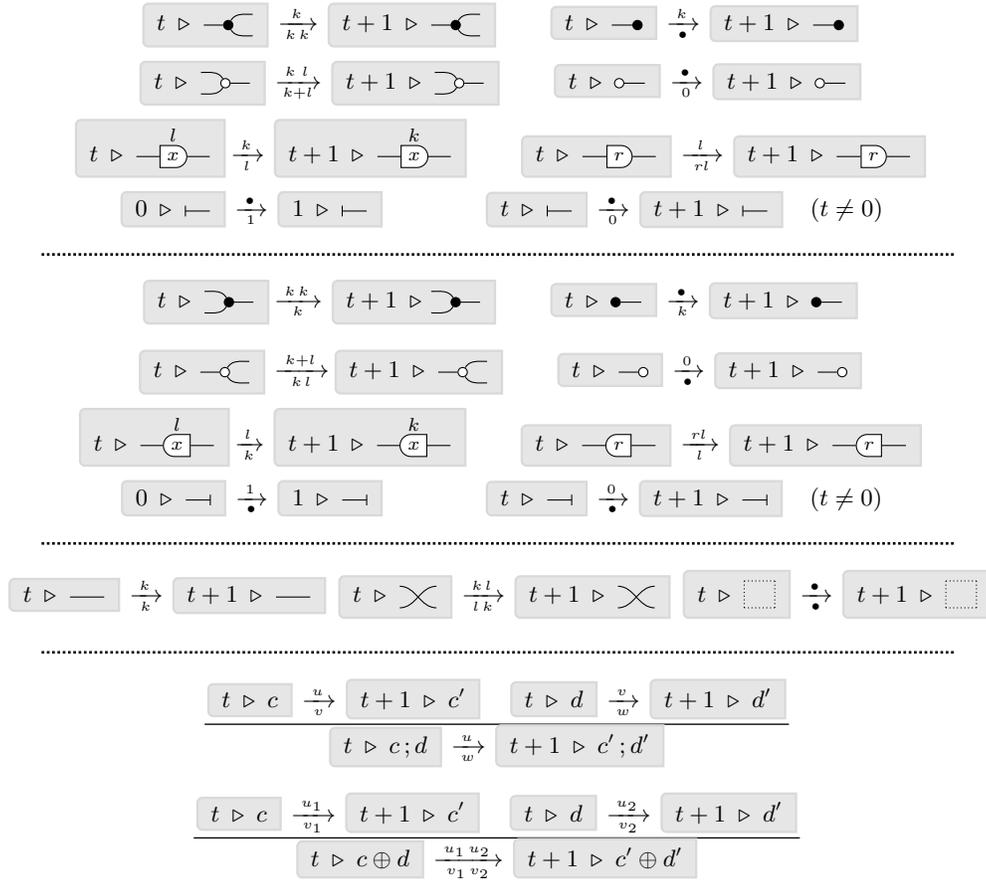

\[
\fullcontext{t}{\BcomultT} \dtrans{\,k\,}{k \labelSep k} \fullcontext{t+1}{\BcomultT}\quad\quad
\fullcontext{t}{\BcounitT} \dtrans{k}{\bullet} \fullcontext{t+1}{\BcounitT}
\]
\[
\fullcontext{t}{\WmultT} \dtrans{k \labelSep\, l }{k+l} \fullcontext{t+1}{\WmultT} \quad\quad
\fullcontext{t}{\WunitT} \dtrans{\bullet}{0} \fullcontext{t+1}{\WunitT}
\]
\[
\fullcontext{t}{\stregister{l}}\! \dtrans{k}{l} \fullcontext{t+1}{\stregister{k}}  \qquad 
\fullcontext{t}{\tscalar{r}}  \dtrans{\,\,l\,\,}{rl} \fullcontext{t+1}{\tscalar{r}}
\]
\[
\fullcontext{0}{\One} \dtrans{\bullet}{1} \fullcontext{1}{\One} \qquad \qquad \fullcontext{t}{\One} \dtrans{\bullet}{0} \fullcontext{t+1}{\One} \quad (t \neq 0)
\]
\hdashrule{\linewidth}{1pt}{1pt}
\[
\fullcontext{t}{\BmultT}\! \dtrans{k \labelSep k}{k} \fullcontext{t+1}{\!\BmultT} \quad\quad
\fullcontext{t}{\BunitT} \dtrans{\bullet}{k} {\fullcontext{t+1}{\BunitT}}
\]
\[
\fullcontext{t}{\WcomultT}\dtrans{ k+l}{k \labelSep l }\! \fullcontext{t+1}{\WcomultT}\quad\quad
\fullcontext{t}{\WcounitT} \dtrans{0}{\bullet} \fullcontext{t+1}{\WcounitT}
\]
\[
\fullcontext{t}{\stcoregister{l}}\! \dtrans{l}{k} \fullcontext{t+1}{\stcoregister{k}}
\qquad
\fullcontext{t}{\tcoscalar{r}}  \dtrans{rl}{\,l\,} \fullcontext{t+1}{\tcoscalar{r}}
\]
\[
\fullcontext{0}{\coOne} \dtrans{1}{\bullet} \fullcontext{1}{\coOne} \qquad \qquad \fullcontext{t} {\coOne} \dtrans{0}{\bullet} \fullcontext{t+1}{\coOne} \quad (t \neq 0)
\]
\hdashrule{\linewidth}{1pt}{1pt}
\[\mathclap{
\fullcontext{t}{\IdnetT}  \dtrans{k}{k} \fullcontext{t+1}{\IdnetT} \;
\fullcontext{t}{\symNetT} \dtrans{k \labelSep l}{l \labelSep k} \fullcontext{t+1}{\symNetT}\;
\fullcontext{t}{\ZeronetT}\dtrans{\bullet}{\bullet}\fullcontext{t+1}{\ZeronetT}
}\]
\hdashrule{\linewidth}{1pt}{1pt}
\[
\derivationRule{\fullcontext{t}{c}\dtrans{u}{v} \fullcontext{t+1}{c'}\quad \fullcontext{t}{d}\dtrans{v}{w} \fullcontext{t+1}{d'} 
}
{\fullcontext{t}{c\poi d} \dtrans{u}{w} \fullcontext{t+1}{c'\poi d'}}{} 
\]
\[
\derivationRule{\fullcontext{t}{c}\dtrans{u_1}{v_1} \fullcontext{t+1}{c'}\quad \fullcontext{t}{d}\dtrans{u_2}{v_2} \fullcontext{t+1}{d'} 
}
{\fullcontext{t}{c\tns d} \dtrans{u_1 \labelSep u_2}{v_1 \labelSep v_2} \fullcontext{t+1}{c'\tns d'}}{} 
\]
\caption{Structural rules for operational semantics, with $p\in\Z$, $k,l$ ranging over $\field$ and $u,v,w$ vectors of elements of $\field$ of the appropriate size. The only vector of $\field^0$ is written as $\bullet$ (as in Definition \ref{def:SemIB}), while a vector $(k_1 \; \dots \; k_n)^T \in \field^n$ as $k_1 \dots k_n$\label{fig:operationalSemantics}.}
\end{figure*}

Here we give the structural operational semantics of affine circuits, building on previous work~\cite{Bonchi2015} that considered only the core linear fragment, $\CD$. We consider circuits to be \emph{programs} that have an observable behaviour. Observations are possible interactions at the circuit's interface. Since there are two interfaces: a left and a right, each transition has two labels.

In a transition $\fullcontext{t}{c}\dtrans{v}{w} \fullcontext{t'}{c'}$, $c$ and $c'$ are \emph{states}, that is, circuits augmented with information about which values $k \in \field$ are stored in each register ($\circuitXT$ and $\circuitXopT$) at that instant of the computation. When transitioning to $c'$,
the $v$
above  
the arrow is a vector of values with which $c$ synchronises on the left, 
and the $w$ below the arrow accounts for the synchronisation on the right.
States are decorated with runtime contexts: $t$ and $t'$ are (possibly negative) integers that---intuitively---indicate the time when the transition happens. 
Indeed, in~\figref{fig:operationalSemantics}, every rule advances time by $1$ unit. 
``Negative time'' is important: as we shall see in Example~\ref{example:1overX}, some executions must start in the past.

The rules in the top section of \figref{fig:operationalSemantics} provide the semantics for the generators in \eqref{eq:SFcalculusSyntax1}: $\Bcomult$ is a \emph{copier}, duplicating the signal arriving on the left; $\Bcounit$ accepts any signal on the left and discards it, producing nothing on the right; $\Wmult$ is an \emph{adder} that takes two signals on the left and emits their sum on the right, $\Wunit$ emits the constant $0$ signal on the right; $\scalar$ is an \emph{amplifier}, multiplying the signal on the left by the scalar $k\in \field$. All the generators described so far are stateless. State is provided by
$\stregister{l}$
which is a \emph{register}; a synchronous one place buffer with the value $l$ stored. When it receives some value $k$ on the left, it emits $l$ on the right and stores $k$. The behaviour of the affine generator $\One$ depends on the time: when $t=0$, it emits $1$, otherwise it emits $0$. Observe that the behaviour of all other generators is time-independent.

So far, we described the behaviour of the components in \eqref{eq:SFcalculusSyntax1} using the intuition that signal flows from left to right: in a transition $\dtrans{v}{w}$, the signal $v$ on the left is thought as trigger and $w$ as effect. For the generators in \eqref{eq:SFcalculusSyntax2}, whose behaviour is defined by  the rules in the second section of \figref{fig:operationalSemantics}, the behaviour is symmetric---indeed, here it is helpful to think of signals as flowing from right to left. The next section of~\figref{fig:operationalSemantics} specifies the behaviours of the structural connectors of~\eqref{eq:SFcalculusSyntax3}: $\symNetT$ is a \emph{twist}, swapping two signals, $\ZeronetT$ is the empty circuit and $\IdnetT$ is the \emph{identity} wire: the signals on the left and on the right ports are equal. Finally, the rule for sequential $\poi$ composition forces the two components to have the same value $v$ on the shared interface, while for parallel $\tns$ composition, components can proceed independently. Observe that both forms of composition require component transitions to happen at the same time.

\begin{definition} Let $c\in \ACD$. The \emph{initial state} $c_0$ of $c$ is the one where all the registers store $0$.
A \emph{computation} of $c$ starting at time $t\leq 0$ is a (possibly infinite) sequence of transitions 
\begin{equation}\label{eq:computation}
\fullcontext{t}{c_0} \dtrans{v_t}{w_t} \fullcontext{t+1}{c_{1}}  \dtrans{v_{t+1}}{w_{t+1}} \fullcontext{t+2}{c_{2}} \dtrans{v_{t+2}}{w_{t+2}} \dots
\end{equation}
\end{definition}
Since 
all transitions increment the time by $1$, it suffices to record the time at which a computation starts. As a result, to simplify notation, we will omit the runtime context after the first transition and, instead of~\eqref{eq:computation}, write
\[\context{t}c_0 \dtrans{v_t}{w_t} c_{1}  \dtrans{v_{t+1}}{w_{t+1}} c_{2} \dtrans{v_{t+2}}{w_{t+2}} \dots\]

\begin{example}\label{exm:opsem}
The circuit in Example \ref{ex:loop} can perform the following computation. 
\begin{multline*}
\context{0} \!{
\tikzset{x=1em, y=2.1ex}
\InputIfFileExists{accu-buffer-0.tikz}{}{\input{./tikz/accu-buffer-0.tikz}}
\tikzset{x=1em, y=1.5ex}
}\!\!\! \dtrans{1}{1} 
 \!{
\tikzset{x=1em, y=2.1ex}
\InputIfFileExists{accu-buffer-1.tikz}{}{\input{./tikz/accu-buffer-1.tikz}}
\tikzset{x=1em, y=1.5ex}
}\!\!\! \dtrans{0}{1} 
 \!{
\tikzset{x=1em, y=2.1ex}
\InputIfFileExists{accu-buffer-1.tikz}{}{\input{./tikz/accu-buffer-1.tikz}}
\tikzset{x=1em, y=1.5ex}
}\!\!\! \dtrans{0}{1} \cdots
\end{multline*}
\end{example}
In the example above, the flow has a clear left-to-right orientation, albeit with a feedback loop. For arbitrary circuits of $\ACD$ this is not always the case, which sometimes results in unexpected operational behaviour.

\begin{example}\label{example:1overX}
In $
\tikzset{x=1em, y=2.1ex}
\InputIfFileExists{one-coreg.tikz}{}{\input{./tikz/one-coreg.tikz}}
\tikzset{x=1em, y=1.5ex}
$ is not possible to identify a consistent flow: $\One$ goes from left to right, while $\circuitXopT$ from right to left. Observe that there is no computation starting at $t=0$, since in the initial state the register contains $0$ while $\One$ must emit $1$. There is, however, a (unique!) computation starting at time $t=-1$, that loads the register with $1$ before $\One$ can also emit $1$ at time $t=0$.
\[\context{-1} \onecoreg{0}
 \dtrans{\bullet}{1} \onecoreg{1} \dtrans{\bullet}{0} \onecoreg{0} \dtrans{\bullet}{0} \onecoreg{0} \dtrans{\bullet}{0} \dots\]
Similarly, $\onetwocoregs{}{}$ features a unique computation starting at time $t=-2$.
\[
\context{-2} \onetwocoregs{0}{0}
 \dtrans{\bullet}{1} \onetwocoregs{0}{1} \dtrans{\bullet}{0} \onetwocoregs{1}{0} \dtrans{\bullet}{0} \onetwocoregs{0}{0} \dtrans{\bullet}{0} \dots
\] 
\end{example}

It is worthwhile clarifying the reason why, in the affine calculus, some computations start in the past. As we have already mentioned, in the linear fragment the semantics of all generators is time-independent. It follows easily that time-independence is a property enjoyed by all purely linear circuits.  The behaviour of $\One$, however, enforces a particular action to occur at time 0. Considering this in conjunction with a right-to-left register results in $\onecoreg{}$, and the effect is to anticipate that action by one step to time -1, as shown in Example~\ref{example:1overX}. It is obvious that this construction can be iterated, and it follows that the presence of a single time-dependent generator results in a calculus in which the computation of some terms must start at a finite, but unbounded time in the past.

\begin{example}\label{example:empty}
Another circuit with conflicting flow is $
\tikzset{x=1em, y=2.1ex}
}
\tikzset{x=1em, y=1.5ex}
$.
Here there is no possible transition at $t=0$, since at that time $\One$ must emit a $1$ and $\WcounitT$ can only synchronise on a $0$.
Instead, the circuit $\ZeronetT$ can always perform an infinite computation $\context{t} \ZeronetT \dtrans{\bullet}{\bullet}  \ZeronetT \dtrans{\bullet}{\bullet} \dots$, for any $t\leq 0$. Roughly speaking, the computations of these two $\sort{0}{0}$ circuits are operational mirror images of the two possible denotations of Proposition~\ref{prop:twopossibility}. This intuition will be made formal in Section \ref{sec:fullabs}. For now, it is worth observing that for all $c$, $\ZeronetT \tns c$ can perform the same computations of $c$, while $
\tikzset{x=1em, y=2.1ex}
}
\tikzset{x=1em, y=1.5ex}
 \tns c$ cannot ever make a transition at time $0$.
\end{example}

\begin{example}\label{example:spancospan}
Consider the circuit $\spanregs{}$, which again features conflicting flow. Our equational theory equates it with $\IdnetT$, but the computations involved are subtly different. Indeed, for any sequence $a_i\in \field$, it is obvious that $\IdnetT$ admits the computation
\begin{equation}\label{eq:idcomp}
\context{0} \IdnetT \dtrans{a_0}{a_0} \IdnetT \dtrans{a_1}{a_1} \IdnetT \dtrans{a_2}{a_2} \dots
\end{equation}
The circuit $\spanregs{}$ admits a similar computation, but we must begin at time $t=-1$ in order to first ``load'' the registers with $a_0$:
\begin{equation}\label{eq:spanregscomp}
\context{-1} \spanregs{0} \dtrans{0}{0}  \spanregs{a_0} \dtrans{a_0}{a_0} \spanregs{a_1} \dtrans{a_1}{a_1} \spanregs{a_2}  \dtrans{a_2}{a_2} \dots
\end{equation} 
The circuit $\cospanregs{}{}$, which again is equated with $\IdnetT$ by the equational theory, is more tricky.
Although every computation of $\IdnetT$ can be reproduced, $\cospanregs{}{}$ admits additional, problematic computations. Indeed, consider
\begin{equation}\label{eq:cospanregscomp}
\context{0} \cospanregs{0}{0} \dtrans{0}{1}  \cospanregs{0}{1}
\end{equation}
at which point no further transition is possible---the circuit can deadlock. 
\end{example}

The following lemma is an easy consequence of the rules of \figref{fig:operationalSemantics} and follows by structural induction. It states that all circuits can stay idle \emph{in the past}.
\begin{lemma}\label{lemma:idle}
Let $c\in \ACD[n,m]$ with initial state $c_0$.
Then $\context{t} c_0 \dtrans{0}{0} c_0$ if $t < 0$.
\end{lemma}

\subsection{Trajectories}

For the non-affine version of the signal flow calculus, we studied in \cite{Bonchi2015} \emph{traces} arising from computations. For the affine extension, this is not possible since, as explained above, we must also consider computations that start in the past. In this paper, rather than traces we adopt a common control theoretic notion.

\begin{definition}
An $(n,m)$-\emph{trajectory} $\sigma$ is a $\Z$-indexed sequence $\sigma \from \Z \to \field^n\times\field^m$ that is finite in the past, i.e., for which $\exists j\in \Z$ such that $\sigma(i) = (0,0)$ for $i\leq j$.
\end{definition}
By the universal property of the product we can identify $\sigma \from \Z \to \field^n\times\field^m$ with the pairing $\langle\sigma_l, \sigma_r\rangle$ of  $\sigma_l \from \Z \to \field^n$ and $\sigma_r \from \Z \to\field^m$. A $(k,m)$-trajectory $\sigma$ and $(m,n)$-trajectory $\tau$ are \emph{compatible} if $\sigma_r = \tau_l$. In this case,
we can define their composite, a $(k,n)$-trajectory $\sigma\poi \tau$ by $\sigma\poi \tau \df \langle \sigma_l, \tau_r \rangle$.
Given an $(n_1,m_1)$-trajectory $\sigma_1$, 
and an $(n_2,m_2)$-trajectory $\sigma_2$, their product, an $(n_1+n_2,m_1+m_2)$-trajectory 
$\sigma_1\tns \sigma_2$, is defined $(\sigma_1\tns \sigma_2)(i)\df \left(\begin{array}{c}
				  \!\sigma(i)\! \\
				  \!\tau(i)\!
				\end{array}\right)$. 
Using these two operations we can organise \emph{sets} of trajectories into a prop.
\begin{definition}\label{def:opsettraj}
The composition of two sets of trajectories is defined 
as $
S\poi T \df\{\sigma\poi \tau\mid \sigma\in S, \tau\in T \text{ are compatible}\}.
$
The product of sets of trajectories is defined 
as $
S_1\tns S_2\df \{ \sigma_1\tns \sigma_2 \mid  \sigma_1\in S_1, \sigma_2\in S_2\}.
$ 
\end{definition}
Clearly both operations are strictly associative. The unit for $\tns$ is the singleton with the unique $(0,0)$-trajectory.
Also $\poi$ has a two sided identity, given by sets of ``copycat'' $(n,n)$-trajectories. Indeed, we have that:
\begin{proposition}
Sets of $(n,m)$-trajectories are the arrows $n\rightarrow m$ of a prop $\Traj$ with composition and monoidal product given as in Definition \ref{def:opsettraj}. 
\end{proposition}

$\Traj$ serves for us as the domain for operational semantics: given a circuit $c$ and an \emph{infinite} computation
$$\context{t} c_0 \dtrans{u_t}{v_t}c_1\dtrans{u_{t+1}}{v_{t+1}}c_2\dtrans{u_{t+2}}{v_{t+2}} \dots $$
its associated trajectory $\sigma$ is 
\begin{equation}\label{eq:compt-traj}
\sigma(i)= \begin{cases}
          (u_{i},v_{i}) & \text{ if } i \geq t,\\
          (0 ,0)             & \text{ otherwise.}\\   \end{cases}
\end{equation}

\vspace{-.2cm}

\begin{definition}\label{def:trajectories}
For a circuit $c$, $\osem{c}$ is the set of trajectories given by its infinite computations, following 
the translation~\eqref{eq:compt-traj} above.
\end{definition}
The assignment $c\mapsto\osem{c}$ is compositional, that is:
\begin{theorem}\label{thm:opsem-morphism}
$\osem{\cdot}\colon \ACD \to \Traj$ is a morphism of props.
\end{theorem}
\begin{proof}In Appendix~\ref{app:proofs}.\end{proof}

\begin{example}
Consider the computations \eqref{eq:idcomp} and \eqref{eq:spanregscomp} from Example \ref{example:spancospan}. According to~\eqref{eq:compt-traj} both are translated into the trajectory $\sigma$ mapping $i\geq 0$ into $(a_i,a_i)$ and $i < 0$ into $(0,0)$. The reader can easily verify that, more generally, it holds that $\osem{\IdnetT}=\osem{\spanregs{}}$. At this point it is worth to remark that the two circuits would be distinguished when looking at their traces: the trace of computation \eqref{eq:idcomp} is different from the trace of \eqref{eq:spanregscomp}. Indeed, the full abstraction result in \cite{Bonchi2015} does not hold for all circuits, but only for those of a certain kind. The affine extension obliges us to consider computations that starts in the past and, in turn, this drives us toward a stronger full abstraction result, shown in the next section.

Before concluding, it is important to emphasise that $\osem{\IdnetT}=\osem{\cospanregs{}{}}$ also holds. Indeed, problematic computations, like \eqref{eq:cospanregscomp}, are all finite and, by definition, do not give rise to any trajectory.
The reader should note that the use of trajectories is not a semantic device to get rid of problematic computations. In fact, trajectories do not appear in the statement of our full abstraction result; they are merely a convenient tool to prove it.
Another result (Proposition \ref{prop:sfg-infinite}) independently takes care of ruling out problematic computations. 
\end{example}


\section{Contextual Equivalence and Full Abstraction}\label{sec:fullabs}

This section contains the main contribution of the paper: a traditional full abstraction result asserting that contextual equivalence agrees with denotational equivalence. It is not a coincidence that we prove this result in the affine setting: affinity plays a crucial role, both in its statement and proof. In particular, Proposition~\ref{prop:twopossibility} gives us two possibilities for the denotation of $\sort{0}{0}$ circuits: \textit{(i)} $\emptyset$---which, roughly speaking, means that there is a problem (see e.g. Example~\ref{example:empty}) and no infinite computation is possible---or \textit{(ii)} $id_0$, in which case infinite computations are possible. This provides us with a basic notion of observation, akin to observing termination vs non-termination in the $\lambda$-calculus. 
 
\begin{definition}
For a circuit $c\in\ACD[0,0]$ we write $c\infcomp$ if $c$ can perform an infinite computation and $c \ninfcomp$ otherwise. For instance $\ZeronetT \infcomp$, while $
\tikzset{x=1em, y=2.1ex}
}
\tikzset{x=1em, y=1.5ex}
\ninfcomp$.
\end{definition}

To be able to make observations about arbitrary circuits we need to introduce an appropriate notion of context.
Roughly speaking, contexts for us are $\sort{0}{0}$-circuits with a hole into which we can plug another circuit. 
Since ours is a variable-free presentation, ``dangling wires'' assume the role of free variables~\cite{GhicaL17}: restricting to $\sort{0}{0}$ contexts is therefore analogous to considering \emph{ground} contexts---i.e. contexts with no free variables---a standard concept of programming language theory.

To define contexts formally, we extend the syntax of~Section \ref{sec:syntax} with an extra generator ``$-$'' of sort $\sort{n}{m}$. A $\sort{0}{0}$-circuit of this extended syntax is a \emph{context} when ``$-$'' occurs exactly once. Given an $\sort{n}{m}$-circuit $c$ and a context $C[-]$, we write $C[c]$ for the circuit obtained by replacing the unique occurrence of ``$-$''  by $c$.

With this setup, given an $\sort{n}{m}$-circuit $c$, we can insert it into a context $C[-]$ and observe the possible outcome: either $C[c]\infcomp$ or $C[c]\ninfcomp$. This naturally leads us to contextual equivalence and the statement of our main result.

\begin{definition}
Given $c,d\in \ACD[n,m]$, we say that they are \emph{contextually equivalent}, written $c \Obseq d$, if for all contexts $C[-]$, $$C[c]\infcomp \text{ iff } C[d]\infcomp\text{.}$$
\end{definition}

\begin{example}
Recall from Example~\ref{example:spancospan}, the circuits $\IdnetT$ and $\cospanregs{}{}$. 
Take the context $C[-]=c_\sigma \poi \; - \; \poi c_\tau$ for $c_\sigma \in \ACD[0,1]$ and $c_\tau \in \ACD[1,0]$. Assume that $c_\sigma$ and $c_\tau$ have a single infinite computation. Call $\sigma$ and $\tau$ the corresponding trajectories. 
If $\sigma = \tau$, both $C[\IdnetT]$ and $C[\cospanregs{}{}]$ would be able to perform an infinite computation. Instead if $\sigma \neq \tau$, none of them would perform any infinite computation: $\IdnetT$ would stop at time $t$, for $t$ the first moment
such that $\sigma(t)\neq \tau(t)$, while $C[\cospanregs{}{}]$ would stop at time $t+1$.

Now take as context $C[-] = \BunitT\poi -\poi \BcounitT$. In  contrast to $c_\sigma$ and $c_\tau$, $\BunitT$ and $\BcounitT$ can perform more than one single computation: at any time they can nondeterministically emit any value. Thus every computation of $C[\IdnetT] = 
\tikzset{x=1em, y=2.1ex}
\begin{tikzpicture}
	\begin{pgfonlayer}{nodelayer}
		\node [style=black] (0) at (0.5, -0) {};
		\node [style=black] (1) at (-0.75, -0) {};
	\end{pgfonlayer}
	\begin{pgfonlayer}{edgelayer}
		\draw (0) to (1);
	\end{pgfonlayer}
\end{tikzpicture}}
\tikzset{x=1em, y=1.5ex}
$ can \emph{always} be extended to an infinite one, forcing synchronisation of $\BunitT$ and $\BcounitT$ at each step. For $C[\cospanregs{}{}]=
\tikzset{x=1em, y=2.1ex}
\begin{tikzpicture}
	\begin{pgfonlayer}{nodelayer}
		\node [style=reg] (0) at (-0.25, 0) {$x$};
		\node [style=coreg] (1) at (1.25, 0) {$x$};
		\node [style=black] (2) at (-1.75, 0) {};
		\node [style=black] (4) at (2.75, 0) {};
	\end{pgfonlayer}
	\begin{pgfonlayer}{edgelayer}
		\draw (2) to (0);
		\draw (0) to (1);
		\draw (1) to (4);
	\end{pgfonlayer}
\end{tikzpicture}
}
\tikzset{x=1em, y=1.5ex}
$, $\BunitT$ and $\BcounitT$ may emit different values at time $t$, but the computation will get stuck at $t+1$. However, our definition of $\infcomp$ only cares about whether $C[\cospanregs{}{}]$ \emph{can} perform an infinite computation. Indeed it can, as long as $\BunitT$ and $\BcounitT$ consistently emit the same value at each time step. 

If we think of contexts as tests, and say that a circuit $c$ passes test $C[-]$ if $C[c]$ perform an infinite computation, then our notion of contextual equivalence is \emph{may-testing} equivalence~\cite{de1984testing}. From this perspective, $\IdnetT$ and $\cospanregs{}{}$ are not \emph{must equivalent}, since the former must pass the test $\BunitT\poi -\poi \BcounitT$ while $\cospanregs{}{}$ may not. It is worth to remark here that the distinction between may and must testing will cease to make sense in Section \ref{sec:sfg} where we identify a certain class of circuits equipped with a proper flow directionality and thus a deterministic, input-output, behaviour.
\end{example}

\begin{theorem}[Full abstraction]\label{thm:fullabstraction}
$c \Obseq d$ iff $c \eqAIH d$
\end{theorem}

The remainder of this section is devoted to the proof of Theorem~\ref{thm:fullabstraction}.
We will start by clarifying the relationship between fractions of polynomials (the denotational domain) and trajectories (the operational domain). 

\subsection{From Polynomial Fractions to Trajectories}
\label{sec:frpoly-traj}
The missing link between polynomial fractions and trajectories are \emph{(formal) Laurent series}: we now recall this notion. Formally, a Laurent series is a function $\sigma\colon \Z \to \field$ for which there exists $j \in \Z$ such that $\sigma(i) = 0$ for all $i < j$. We write $\sigma$ as $\dots,\sigma(-1),\underline{\sigma(0)}, \sigma(1), \dots$ with position $0$ underlined, or as formal sum $\LSum{d}\sigma(i)x^i$. Each Laurent series $\sigma$ has then a \emph{degree} $d \in \Z$, which is the first non-zero element. Laurent series form a field $\laur$: sum is pointwise, product is by convolution, and the inverse $\sigma^{-1}$ of $\sigma$ with degree $d$ is defined as:
\begin{eqnarray}\label{eq:inverse}
\sigma^{-1}(i) = \,
\begin{cases} 0 & \text{ if } i < -d \\
\sigma(d)^{-1}  &\text{ if } i=-d \\
\frac{\sum_{i=1}^{n} \big( \sigma(d+i) \cdot \sigma^{-1}(-d+n-i)\big)}{-\sigma(d)} & \text{ if } i\!=\!-d\!+\!n \text{ for } n\!>\!0
\end{cases}\hspace{-0.7cm}
\end{eqnarray}
Note (formal) power series, which form `just' a ring $\fps$, are a particular case of Laurent series, namely those $\sigma$s for which $d \geq 0$. What is most interesting for our purposes is how polynomials and fractions of polynomials relate to $\laur$ and $\fps$. First, the ring $\poly$ of polynomials embeds into $\fps$, and thus into $\laur$: a polynomial $p_0+p_1x+\dots +p_nx^n$ can also be regarded as the power series $\LSum{0}p_ix^i$ with $p_i=0$ for all $i>n$. Because Laurent series are closed under division, this immediately gives also an embedding of the field of polynomial fractions $\frpoly$ into $\laur$. Note that the full expressiveness of $\laur$ is required: for instance, the fraction $\frac{1}{x}$ is represented as the Laurent series $\dots ,0,1,\underline{0},0, \dots$, which is not a power series, because a non-zero value appears before position $0$. In fact, fractions that are expressible as power series are precisely the \emph{rational} fractions, i.e. of the form  $\frac{k_0+k_1x+k_2x^2 \dots + k_nx^n}{l_0+l_1x+l_2x^2 \dots + l_nx^n}$ where $l_0\neq 0$. 

\noindent \begin{minipage}[c]{.60\linewidth}
Rational fractions form a ring $\ratio$ which, differently from the full field $\frpoly$, embeds into $\fps$. Indeed, whenever $l_0\neq 0$, the inverse of $l_0+l_1x+l_2x^2 \dots + l_nx^n$ is, by \eqref{eq:inverse}, a \emph{bona fide} power series. The commutative diagram on the right is a summary. \end{minipage}
\begin{minipage}[c]{.40\linewidth}
\[
\xymatrix@R=1pt@C=1pt{
\fps \ar@{^{(}->}[rrrr] &  & & & \laur\\
\\
\\
& & \ratio \ar@{_{(}->}[lluuu] \ar@{^{(}->}[rrd]
\\
\poly \ar@{^{(}->}[urr] \ar@{_{(}->}[rrrr] \ar@{^{(}->}[uuuu] & & & & \frpoly \ar@{_{(}->}[uuuu]
}
\]
\end{minipage}

\medskip

Relations between $\laur$-vectors organise themselves into a prop $\ARel{\laur}$ (see Definition \ref{DEF_SV}). There is an evident prop morphism $\iota \colon \ARel{\frpoly} \to \ARel{\laur}$: it maps the empty affine relation on $\frpoly$ to the one on $\laur$, and otherwise applies pointwise the embedding of $\frpoly$ into $\laur$. 
For the next step, observe that trajectories are in fact rearrangements of Laurent series: each pair of vectors $(u,v)\in\laur^n \times \laur^m$, as on the left below, yields the trajectory $\kappa(u,v)$ defined for all $i\in \Z$ as on the right below.
$$(u,v) = \left(\!{\small\left( \begin{array}{c}
				  \!\alpha^1\! \\
				  \vdots\\
				  \!\alpha^n\!
				\end{array}\right)}, \,
{\small\left(%
				\begin{array}{c}
				  \!\beta^1\! \\
				  \vdots\\
				  \!\beta^m\!
				\end{array}\right)}\!\right)
\qquad				
\qquad
\kappa(u,v)(i)=\left(\!{\small\left( \begin{array}{c}
				  \!\alpha^1(i)\! \\
				  \vdots\\
				  \!\alpha^n(i)\!
				\end{array}\right)}, \,
{\small\left(%
				\begin{array}{c}
				  \!\beta^1(i)\! \\
				  \vdots\\
				  \!\beta^m(i)\!
				\end{array}\right)}\!\right)$$
Similarly to $\iota$, the assignment $\kappa$ extends to sets of vectors, and also to a prop morphism from $\ARel{\laur}$ to $\Traj$. Together, $\kappa$ and $\iota$ provide the desired link between operational and denotational semantics.

\begin{theorem}\label{thm:main}
$\osem{\cdot}=\kappa \circ \iota \circ \dsem{\cdot}$
\end{theorem}
\begin{proof}
Since both are symmetric monoidal functors from a free prop, it is enough to check the statement for the generators of $\ACD$. We show, as an example, the case of $\Bcomult$. By Definition \ref{def:SemIB}, $\dsem{\Bcomult} = \left\{\left( p ,\left(\begin{array}{c}
				 \! p\! \\
				 \! p\!
				\end{array}\right) \right) \mid p\in \frpoly \right\}  
$. This  is mapped by $\iota$ to $\left\{\left( \alpha ,\left(\begin{array}{c}
				 \! \alpha\! \\
				 \! \alpha\!
				\end{array}\right) \right) \mid \alpha\in \laur \right\}$. Now, to see that $\kappa(\iota(\dsem{\Bcomult}))= \osem{\Bcomult}$, it is enough to observe that a trajectory $\sigma$ is in $\kappa(\iota(\dsem{\Bcomult}))$ precisely when, for all $i$, there exists some $k_i\in \field$ such that $\sigma(i)=\left( k_i ,\left(\begin{array}{c}
				 \! k_i\! \\
				 \!k_i\!
				\end{array}\right) \right)$.
\end{proof}

\subsection{Proof of Full Abstraction}

We now have the ingredients to prove Theorem \ref{thm:fullabstraction}. First, 
we prove an adequacy result for $\sort{0}{0}$ circuits.

\begin{proposition}\label{cor:dsemexperiments}
Let $c\in \ACD[0,0]$. Then
$\dsem{c}=id_0$ if and only if $c\infcomp$.
\end{proposition}
\begin{proof}
By Proposition \ref{prop:twopossibility}, either $\dsem{c}=id_0$ or $\dsem{c}=\emptyset$, which, combined with Theorem \ref{thm:main}, means that  $ \osem{c}=\kappa \circ \iota (id_0)$ or $\osem{c}=\kappa \circ \iota(\emptyset)$. By definition of $\iota$ this implies that either $\osem{c}$ contains a trajectory or not. In the first case $c\infcomp$; in the second $c \ninfcomp$.
\end{proof}

Next we obtain a result that relates denotational equality in all contexts to equality in $\AIH$.
Note that it is not trivial: since we consider ground contexts it does not make sense to merely consider  
``identity'' contexts. Instead, it is at this point that we make another crucial use of affinity, taking advantage
of the increased expressivity of affine circuits, as showcased by Proposition \ref{prop:denotation-context}.
\begin{proposition}\label{lemma:contextequivimpliesdenequiv}
If $\dsem{C[c]}=\dsem{C[d]}$ for all contexts $C[-]$, then $c \eqAIH d$.
\end{proposition}
\begin{proof}
Suppose that $c \stackrel{\AIH}{\neq} d$. Then $\dsem{c}\neq \dsem{d}$. Since both $\dsem{c}$ and $\dsem{d}$ are affine relations over $\frpoly$, there exists a pair of vectors $(u,v)\in  \frpoly^n \times  \frpoly^m$ that is in one of $\dsem{c}$ and $\dsem{d}$, but not both. Assume wlog that $(u,v)\in\dsem{c}$ and $(u,v)\notin\dsem{d}$. 
By Proposition \ref{prop:denotation-context}, there exists $c_u$ and $c_v$ such that $\dsem{c_u \poi c \poi c_v} = \dsem{c_u} \poi \dsem{c} \poi \dsem{c_v}=\{(\bullet, u)\}\poi\dsem{c} \poi \{(v,\bullet)\}$. Since $(u,v)\in \dsem{c}$, then $\dsem{c_u \poi c \poi c_v}=\{(\bullet, \bullet)\}$. Instead, since $(u,v)\notin \dsem{d}$, we have that $\dsem{c_u \poi d \poi c_v}=\emptyset$. 
Therefore, for the context $C[-]=c_u \poi - \poi c_v\text{,}$ we have that $\dsem{C[c]} \neq \dsem{C[d]}$.
\end{proof}
The proof of our main result is now straightforward.
\begin{proof}[Proof of Theorem \ref{thm:fullabstraction}] 
	Let us first suppose that $c \eqAIH d$. Then $\dsem{C[c]}=\dsem{C[d]}$ for  all contexts $C[-]$, since $\dsem{\cdot}$ is a morphism of props. By Corollary \ref{cor:dsemexperiments}, it follows immediately that $C[c]\infcomp$ if and only if $C[d]\infcomp$, namely $c \Obseq d$.

Conversely, suppose that, for all $C[-]$, $C[c]\infcomp$ iff $C[d]\infcomp$. Again by Corollary \ref{cor:dsemexperiments}, we have that $\dsem{C[c]}=\dsem{C[d]}$. We conclude by invoking Proposition~\ref{lemma:contextequivimpliesdenequiv}.
\end{proof}

\section{Functional Behaviour and Signal Flow Graphs}\label{sec:sfg}


There is a sub-prop $\SFGform$ of $\CD$ of classical \emph{signal flow graphs} (see \emph{e.g.} \cite{mason1953feedback}). 
Here 
signal flows left-to-right, 
 possibly featuring 
\emph{feedback loops}, provided that these go through at least one register. Feedback can be captured algebraically via
an operation
$\Tr{}(\cdot) \: \CD[n+1,m+1] \to \CD[n,m]$ taking $c \: n+1 \to m+1$ to:
\ctikzfig{trace-form}
Following~\cite{Bonchi2015}, let us call $\FC$ the free sub-prop of $\CD$ of circuits built from~\eqref{eq:SFcalculusSyntax3} and the generators of \eqref{eq:SFcalculusSyntax1}, without $\One$. Then $\SFGform$ is defined as the closure of $\FC$ under $\Tr{}(\cdot)$. 
For instance, the circuit of Example~\ref{exm:opsem} is in $\SFGform$.

Signal flow graphs are intimately connected to the executability of circuits. In general, the rules of Figure~\ref{fig:operationalSemantics} 
do not assume a fixed flow orientation.
As a result, some circuits in $\CD$ 
are not executable as \emph{functional input-output} systems, as we have demonstrated with $\onecoreg{}$, $
\tikzset{x=1em, y=2.1ex}
}
\tikzset{x=1em, y=1.5ex}
$ and $\cospanregs{}{}$ of Examples~\ref{example:1overX}-\ref{example:spancospan}. 
Notice that none of these are signal flow graphs. In fact, the circuits of $\SFGform$ 
do not have pathological behaviour, as we shall state more precisely in Proposition~\ref{prop:sfg-infinite}.

At the denotational level, signal flow graphs correspond precisely to \emph{rational} functional behaviours, that is, matrices whose coefficients are in the ring $\ratio$ of \emph{rational fractions} (see Section~\ref{sec:frpoly-traj}).
We call such matrices, rational matrices. One may check that the semantics of a signal flow graph $c \colon\sort{n}{m}$ is always of the form $\dsem{c} = \{(v, A \cdot v) \mid v \in \frpoly^{n} \}$, for some $m \times n$ rational matrix $A$. Conversely, all relations that are the graph of rational matrices can be expressed as signal flow graphs.
\begin{proposition}\label{prop:sfg-rational}
Given $c\colon\sort{n}{m}$, we have $\dsem{c}=\{(p, A\cdot p)\mid p\in\frpoly^n\}$ for some rational $m\times n$ matrix $A$ iff there exists a signal flow graph $f$, i.e., a circuit $f\colon\sort{n}{m}$ of $\SFGform$, such that $\dsem{f}=\dsem{c}$.
\end{proposition}
\begin{proof}
This is a folklore result in control theory which can be found in~\cite{Rutten08_rationalstreamscoalgebraically}. The details of the translation between rational matrices and circuits of $\SFGform$ can be found in~\cite[Section 7]{BonchiSZ17}. 
\end{proof}

The following gives an alternative characterisation of rational matrices---and therefore, by Proposition~\ref{prop:sfg-rational}, of the behaviour of signal flow graphs---that clarifies their role as realisations of circuits.
\begin{proposition}\label{prop:rational-map}
An $m\times n$ matrix is rational iff $A\cdot r\in\ratio^m$ for all $r\in\ratio^n$. 
\end{proposition}
\begin{proof}In Appendix~\ref{app:proofs}.\end{proof}
Proposition~\ref{prop:rational-map} is another guarantee of good behaviour---it justifies the name of inputs (resp. outputs) for the left (resp. right) ports of signal flow graphs. Recall from Section~\ref{sec:frpoly-traj} that rational fractions can be mapped to Laurent series of nonnegative degree, i.e., to plain power series. Operationally, these correspond to trajectories that start after $t=0$. Proposition~\ref{prop:rational-map} guarantees that any trajectory of a signal flow graph whose first nonzero  value on the left appears at time $t=0$, will not have nonzero values on the right starting before time $t=0$. In other words, signal flow graphs can be seen as processing a stream of values from left to right. As a result, their ports can be clearly partitioned into inputs and outputs.

But the circuits of $\SFGform$ are too restrictive for our purposes. For example,  $
\tikzset{x=1em, y=2.1ex}
\InputIfFileExists{x+1.tikz}{}{\input{./tikz/x+1.tikz}}
\tikzset{x=1em, y=1.5ex}
$ 
can also be seen to realise a functional behaviour transforming inputs on the left into outputs on the right yet it is not in $\SFGform$. Its behaviour is no longer linear, but affine. Hence, we need to extend signal flow graphs to include functional affine behaviour. The following definition does just that. 
\begin{definition}\label{def:ASFG}
Let $\ASFG$ be the sub-prop of $\ACD$ obtained from \emph{all} the generators in \eqref{eq:SFcalculusSyntax1}, closed under $\Tr{}(\cdot)$. Its circuits are called \emph{affine signal flow graphs}.
\end{definition}
As before, none of $\onecoreg{}$, $
\tikzset{x=1em, y=2.1ex}
}
\tikzset{x=1em, y=1.5ex}
$ and $\cospanregs{}{}$ from Examples~\ref{example:1overX}-\ref{example:spancospan} are affine signal flow graphs. In fact, $\ASFG$ rules out pathological behaviour: all computations can be extended to be infinite, or in other words, do not get stuck.
\begin{proposition}\label{prop:sfg-infinite}
Given an affine  signal flow graph $f$, for every computation
$$\context{t} f_0 \dtrans{u_t}{v_t}f_{1}\dtrans{u_{t+1}}{v_{p+1}}\dots f_{n} $$
there exists a trajectory $\sigma\in \osem{c}$ such that $\sigma(i) = (u_i,v_i)$ for $t\leq i\leq t+n$.
\end{proposition}
\begin{proof}
By induction on the structure of affine signal flow graphs.
\end{proof}

If $\SFGform$ circuits correspond precisely to $\ratio$-matrices, those of $\ASFG$ correspond precisely to $\ratio$-affine transformations.
\begin{definition}
A map $f:\;\frpoly^n\to \frpoly^m$ is an \emph{affine  map} if there exists an $m\times n$ matrix $A$  and $b\in\frpoly^m$ such that $f(p) = A\cdot p+b$ for all $p\in\frpoly^n$. We call the pair $(A,b)$ the representation of $f$.
\end{definition}
The notion of rational affine map is a straightforward extension of the linear case and so is the characterisation in terms of rational input-output behaviour.
\begin{definition}
An affine map $f:\; p\mapsto A\cdot p+b$ is \emph{rational} if $A$ and $b$ have coefficients in $\ratio$.
\end{definition}
\begin{proposition}\label{prop:affine-rational}
An affine map $f:\;\frpoly^n\to \frpoly^m$ is rational iff $f(r)\in\ratio^m$ for all $r\in\ratio^n$. 
\end{proposition}
\begin{proof}In Appendix~\ref{app:proofs}.\end{proof}
The following extends the correspondence of Proposition~\ref{prop:sfg-rational}, showing that $\ASFG$ is the rightful affine heir of $\SFGform$.
\begin{proposition}\label{prop:asfg-rational}
Given $c\colon\sort{n}{m}$, we have $\dsem{c}=\{(p, f(p))\mid p\in\frpoly^n\}$ for some rational affine map $f$ iff there exists an affine signal flow graph $g$, i.e., a circuit $g\colon\sort{n}{m}$ of $\ASFG$, such that $\dsem{g}=\dsem{c}$.
\end{proposition}
\begin{proof}
Let $f$ be given by $p\mapsto Ap+b$ for some rational $m\times n$ matrix $A$ and vector $b\in\ratio^m$. By Proposition~\ref{prop:sfg-rational}, we can find a circuit $c_A$ of $\SFGform$ such that
\smallskip

\noindent 
\begin{minipage}{.65\textwidth}
 $\dsem{c_A}=\{(p, A\cdot p)\mid p\in\frpoly\}$. Similarly, we can represent $b$ as a signal flow graph $c_b$ of sort $\sort{1}{m}$. Then, the circuit on the right is clearly in $\ASFG$ and verifies $\dsem{c} = \{(p,Ap+b)\mid p\in\frpoly\}$ as required. 
\end{minipage}
\begin{minipage}{.35\textwidth}
\begin{center}
$c\;:=\;
\tikzset{x=1em, y=2.1ex}
\InputIfFileExists{affine-sfg.tikz}{}{\input{./tikz/affine-sfg.tikz}}
\tikzset{x=1em, y=1.5ex}
$
\end{center}
\end{minipage}

For the converse direction it is straightforward to check by structural induction that the denotation of  affine signal flow graphs is the graph (in the set-theoretic sense of pairs of values) of some rational affine map.
\end{proof}

\section{Realisability}\label{sec:realisability}

In the previous section we gave a restricted class of morphisms with good behavioural properties. We may wonder how much of $\ACD$ we can capture with this restricted class. The answer is, in a precise sense: most of it. 

Surprisingly, the behaviours realisable in $\CD$---the purely linear fragment---are not more expressive. In fact, from an operational (or denotational, by full abstraction) point of view, $\CD$ is nothing more than jumbled up version of $\SFGform$. Indeed, it turns out that $\CD$ enjoys a \emph{realisability} theorem: any circuit $c$ of $\CD$ can be associated with one of $\SFGform$, that implements or realises the behaviour of $c$ into an executable form.
\smallskip

\noindent 
\begin{minipage}{.7\textwidth}
But the corresponding realisation may not flow neatly from left to right like  signal flow graphs do---its inputs and outputs may have been moved from one side to the other. Consider for example, the circuit on the right
\end{minipage}
\begin{minipage}{.3\textwidth}
$\quad
\tikzset{x=1em, y=2.1ex}
\InputIfFileExists{sfg-jumbled.tikz}{}{\input{./tikz/sfg-jumbled.tikz}}
\tikzset{x=1em, y=1.5ex}
$
\end{minipage}
\smallskip

\noindent It does not belong to $\SFGform$ but it can be read as a signal flow graph with an input that has been bent and moved to the bottom right. The behaviour it realises can therefore executed by rewiring this port to obtain a signal flow graph:
\[
\tikzset{x=1em, y=2.1ex}
\InputIfFileExists{sfg-rewired.tikz}{}{\input{./tikz/sfg-rewired.tikz}}
\tikzset{x=1em, y=1.5ex}
 \quad\eqAIH\quad 
\tikzset{x=1em, y=2.1ex}
\InputIfFileExists{sfg-rewired-1.tikz}{}{\input{./tikz/sfg-rewired-1.tikz}}
\tikzset{x=1em, y=1.5ex}
\]
We will not make this notion of rewiring precise here but refer the reader to~\cite{Bonchi2015} for the details. The intuition is simply that a rewiring partitions the ports of a circuit into two sets---that we call inputs and outputs---and uses $
\tikzset{x=1em, y=2.1ex}
\InputIfFileExists{cup.tikz}{}{\input{./tikz/cup.tikz}}
\tikzset{x=1em, y=1.5ex}
$ or $
\tikzset{x=1em, y=2.1ex}
\InputIfFileExists{cap.tikz}{}{\input{./tikz/cap.tikz}}
\tikzset{x=1em, y=1.5ex}
$ to bend input ports to the left and and output ports to the right. The realisability theorem then states that we can always recover a (not necessarily unique) signal flow graph from any circuit by performing these operations.
\begin{theorem}{\cite[Theorem~5]{Bonchi2015}}\label{thm:realisability}
Every circuit in $\CD$ is equivalent to the rewiring of a signal flow graph, called its \emph{realisation}.
\end{theorem}
This theorem allows us to extend the notion of inputs and outputs to all circuits of $\CD$.
\begin{definition}
A port of a circuit $c$ of $\CD$ is an \emph{input} (resp. \emph{output}) port, if there exists a realisation for which it is an input (resp. output).
\end{definition}
Note that, since realisations are not necessarily unique, the same port can be both an input and an output. Then, the realisability theorem (Theorem~\ref{thm:realisability}) says that every port is always an input, an output or both (but never neither).

An output-only port is an output port that is not an input port. Similarly an input-only port in an input port that is not an output port.
\begin{example}
The left port of the register $\circuitX$ is input-only whereas its right port is output-only. In the identity wire, both ports are input and output ports. The single port of $\Wunit$ is output-only ; that of $\Bcounit$ is input-only.
\end{example}

While in the purely linear case, all behaviours are realisable, the general case of $\ACD$ is a bit more subtle. To make this precise, we can extend our definition of realisability to include affine signal flow graphs.
\begin{definition}
A circuit of $\ACD$ is \emph{realisable} if its ports can be rewired so that it is equivalent to a circuit of $\ASFG$.
\end{definition}
\begin{example}
$\One$ is realisable; $\boxtikzfig{one-coreg}{r}\,$ is not. 
\end{example}
Notice that Proposition~\ref{prop:asfg-rational}, gives the following equivalent semantic criterion for realisability. Realisable behaviours are precisely those that map rationals to rationals.
\begin{theorem}\label{thm:Filippo-realisability}
A circuit $c$ is realisable iff its ports can be partitioned into two sets, that we call inputs and outputs, such that the corresponding rewiring of $c$ is an affine rational map from inputs to outputs.
\end{theorem}
We offer another perspective on realisability below: realisable behaviours correspond precisely to those for which the $\One$ constants are connected to inputs of the underlying $\CD$-circuit.
First, notice that, since
\[
\tikzset{x=1em, y=2.1ex}
\InputIfFileExists{one-copy.tikz}{}{\input{./tikz/one-copy.tikz}}
\tikzset{x=1em, y=1.5ex}
 \quad \myeq{1-dup} \quad
\tikzset{x=1em, y=2.1ex}
\InputIfFileExists{one2.tikz}{}{\input{./tikz/one2.tikz}}
\tikzset{x=1em, y=1.5ex}
\quad \text{and}\quad  
\tikzset{x=1em, y=2.1ex}
\begin{tikzpicture}
	\begin{pgfonlayer}{nodelayer}
		\node [style=none] (0) at (-0.75, -0) {};
		\node [style=black] (1) at (0.5, -0) {};
		\node [style=none] (2) at (-0.75, -0.25) {};
		\node [style=none] (3) at (-0.75, 0.25) {};
	\end{pgfonlayer}
	\begin{pgfonlayer}{edgelayer}
		\draw (0.center) to (1);
		\draw (3.center) to (2.center);
	\end{pgfonlayer}
\end{tikzpicture}}
\tikzset{x=1em, y=1.5ex}
 \; \myeq{1-del} \quad
\tikzset{x=1em, y=2.1ex}
\InputIfFileExists{empty-diag.tikz}{}{\input{./tikz/empty-diag.tikz}}
\tikzset{x=1em, y=1.5ex}
\]
in $\AIH$, we can assume without loss of generality that each circuit contains exactly one $\One\,$.
\begin{proposition}\label{prop:single-foot}
Every circuit $c$ of $\ACD$ is equivalent to one with precisely one $\One$ and no $\coOne$.
\end{proposition}
\begin{proof}In Appendix~\ref{app:proofs}.\end{proof}

For $c\colon\sort{n}{m}$ a circuit of $\ACD$, we will call $\hat{c}$ the circuit of $\CD$ of sort $\sort{n+1}{m}$ that one obtains by first transforming $c$ into an equivalent circuit with a single $\One$ and no $\coOne$ as above, then removing this $\One$, and replacing it by an identity wire that extends to the left boundary.
\begin{theorem}\label{thm:affine-realisability}
A circuit $c$ is realisable iff $\One$ is connected to an input port of $\hat{c}$.
\end{theorem}

\section{Conclusion and Future Work}\label{sec:conclusion}

We introduced the operational semantics of the \emph{affine} extension of the signal flow calculus and proved that contextual equivalence coincides with denotational equality, previously introduced and axiomatised in~\cite{BonchiPSZ19}.
We have observed that, at the denotational level, affinity provides two key properties (Propositions~\ref{prop:denotation-context} and~\ref{prop:twopossibility}) for the proof of full abstraction.
However, at the operational level, affinity forces us to consider computations starting in the \emph{past} (Example \ref{example:1overX}) as the syntax allows terms lacking a proper flow directionality. This leads to circuits that might deadlock ($
\tikzset{x=1em, y=2.1ex}
}
\tikzset{x=1em, y=1.5ex}
$ in Example \ref{example:empty}) or perform some problematic computations ($\cospanregs{}{}$ in Example \ref{example:spancospan}). We have identified a proper subclass of circuits, called affine signal flow graphs (Definition \ref{def:ASFG}), that possess an inherent flow directionality:  in these circuits, the same pathological behaviours do not arise (Proposition \ref{prop:sfg-infinite}). This class is not too restrictive as it captures all desirable behaviours: a realisability result (Theorem \ref{thm:Filippo-realisability}) states that all and only the circuits that do not need computations to start in the past are equivalent to (the rewiring of) an affine signal flow graph.  

The reader may be wondering why we do not restrict the syntax to affine signal flow graphs. The reason is that, like in the behavioural approach to control theory \cite{Willems2007}, the lack of flow direction is what allows the (affine) signal flow calculus to achieve a strong form of compositionality and a complete axiomatisation (see~\cite{Bonchi2015} for a deeper discussion).

We expect that similar methods and results can be extended to other models of computation. Our next step is to tackle Petri nets, which, as shown in~\cite{BHPSZ-popl19}, can be regarded as terms of the signal flow calculus, but over $\N$ rather than a field.

\bibliographystyle{splncs04}
\bibliography{affinepaper}

\newpage
\appendix

\section*{Appendix}\label{app:main}
\section{Interacting Hopf Algebras: a Complete Axiomatisation of Affine Circuits}
\label{sec:axiomatisation}

This appendix contains the equational theory of affine relations over a field $\field$, called the theory of Affine Interacting Hopf algebras ($\AIH{\field}$), as it appears in~\cite{interactinghopf} (for the linear fragment) and~\cite{BonchiPSZ19} (for the affine extension).
The axioms are in Figure~\ref{fig:ih}; we briefly explain them below.
\begin{figure*}[htbp]
\begin{align*}

\tikzset{x=1em, y=2.1ex}
\InputIfFileExists{ax/add-associative.tikz}{}{\input{./tikz/ax/add-associative.tikz}}
\tikzset{x=1em, y=1.5ex}
\;\myeq{$\circ$-as}\;\; 
\tikzset{x=1em, y=2.1ex}
\InputIfFileExists{ax/add-associative-1.tikz}{}{\input{./tikz/ax/add-associative-1.tikz}}
\tikzset{x=1em, y=1.5ex}
&\qquad  
\tikzset{x=1em, y=2.1ex}
\InputIfFileExists{ax/add-commutative.tikz}{}{\input{./tikz/ax/add-commutative.tikz}}
\tikzset{x=1em, y=1.5ex}
\myeq{$\circ$-co}\;\;\, 
\tikzset{x=1em, y=2.1ex}
\InputIfFileExists{ax/add.tikz}{}{\input{./tikz/ax/add.tikz}}
\tikzset{x=1em, y=1.5ex}
\qquad  
\tikzset{x=1em, y=2.1ex}
\InputIfFileExists{ax/add-unital-left.tikz}{}{\input{./tikz/ax/add-unital-left.tikz}}
\tikzset{x=1em, y=1.5ex}
\myeq{$\circ$-unl}\;
\tikzset{x=1em, y=2.1ex}
\InputIfFileExists{ax/id.tikz}{}{\input{./tikz/ax/id.tikz}}
\tikzset{x=1em, y=1.5ex}
\\
 
\tikzset{x=1em, y=2.1ex}
\InputIfFileExists{ax/co-add-associative-1.tikz}{}{\input{./tikz/ax/co-add-associative-1.tikz}}
\tikzset{x=1em, y=1.5ex}
\;\;\myeq{$\circ$-coas}\; 
\tikzset{x=1em, y=2.1ex}
\InputIfFileExists{ax/co-add-associative.tikz}{}{\input{./tikz/ax/co-add-associative.tikz}}
\tikzset{x=1em, y=1.5ex}
&\qquad  
\tikzset{x=1em, y=2.1ex}
\InputIfFileExists{ax/co-add-commutative.tikz}{}{\input{./tikz/ax/co-add-commutative.tikz}}
\tikzset{x=1em, y=1.5ex}
\;\;\,\myeq{$\circ$-coco}\;\; 
\tikzset{x=1em, y=2.1ex}
\InputIfFileExists{ax/co-add.tikz}{}{\input{./tikz/ax/co-add.tikz}}
\tikzset{x=1em, y=1.5ex}
\qquad  
\tikzset{x=1em, y=2.1ex}
\InputIfFileExists{ax/co-add-unital-left.tikz}{}{\input{./tikz/ax/co-add-unital-left.tikz}}
\tikzset{x=1em, y=1.5ex}
\;\myeq{$\circ$-counl}\;
\tikzset{x=1em, y=2.1ex}
\InputIfFileExists{ax/id.tikz}{}{\input{./tikz/ax/id.tikz}}
\tikzset{x=1em, y=1.5ex}
 \\
 
\tikzset{x=1em, y=2.1ex}
\InputIfFileExists{ax/copy-associative.tikz}{}{\input{./tikz/ax/copy-associative.tikz}}
\tikzset{x=1em, y=1.5ex}
\;\;\,\myeq{$\bullet$-coas}\;\; 
\tikzset{x=1em, y=2.1ex}
\InputIfFileExists{ax/copy-associative-1.tikz}{}{\input{./tikz/ax/copy-associative-1.tikz}}
\tikzset{x=1em, y=1.5ex}
&\qquad  
\tikzset{x=1em, y=2.1ex}
\InputIfFileExists{ax/copy-commutative.tikz}{}{\input{./tikz/ax/copy-commutative.tikz}}
\tikzset{x=1em, y=1.5ex}
\;\;\,\myeq{$\bullet$-coco}\;\; 
\tikzset{x=1em, y=2.1ex}
\InputIfFileExists{ax/copy.tikz}{}{\input{./tikz/ax/copy.tikz}}
\tikzset{x=1em, y=1.5ex}
\qquad  
\tikzset{x=1em, y=2.1ex}
\InputIfFileExists{ax/copy-unital-left.tikz}{}{\input{./tikz/ax/copy-unital-left.tikz}}
\tikzset{x=1em, y=1.5ex}
\;\myeq{$\bullet$-counl}\;
\tikzset{x=1em, y=2.1ex}
\InputIfFileExists{ax/id.tikz}{}{\input{./tikz/ax/id.tikz}}
\tikzset{x=1em, y=1.5ex}
 \\

\tikzset{x=1em, y=2.1ex}
\InputIfFileExists{ax/co-copy-associative.tikz}{}{\input{./tikz/ax/co-copy-associative.tikz}}
\tikzset{x=1em, y=1.5ex}
\;\myeq{$\bullet$-as}\;\; 
\tikzset{x=1em, y=2.1ex}
\InputIfFileExists{ax/co-copy-associative-1.tikz}{}{\input{./tikz/ax/co-copy-associative-1.tikz}}
\tikzset{x=1em, y=1.5ex}
&\qquad  
\tikzset{x=1em, y=2.1ex}
\InputIfFileExists{ax/co-copy-commutative.tikz}{}{\input{./tikz/ax/co-copy-commutative.tikz}}
\tikzset{x=1em, y=1.5ex}
\,\myeq{$\bullet$-co}\;\, 
\tikzset{x=1em, y=2.1ex}
\InputIfFileExists{ax/co-copy.tikz}{}{\input{./tikz/ax/co-copy.tikz}}
\tikzset{x=1em, y=1.5ex}
\qquad  
\tikzset{x=1em, y=2.1ex}
\InputIfFileExists{ax/co-copy-unital-left.tikz}{}{\input{./tikz/ax/co-copy-unital-left.tikz}}
\tikzset{x=1em, y=1.5ex}
\myeq{$\bullet$-unl}\;
\tikzset{x=1em, y=2.1ex}
\InputIfFileExists{ax/id.tikz}{}{\input{./tikz/ax/id.tikz}}
\tikzset{x=1em, y=1.5ex}

\end{align*}
\hdashrule{\linewidth}{1pt}{1pt}

\begin{equation*}

\tikzset{x=1em, y=2.1ex}
\InputIfFileExists{ax/add-copy-bimonoid.tikz}{}{\input{./tikz/ax/add-copy-bimonoid.tikz}}
\tikzset{x=1em, y=1.5ex}
\;\;\myeq{$\circ\bullet$-bi}\;\;
\tikzset{x=1em, y=2.1ex}
\InputIfFileExists{ax/add-copy-bimonoid-1.tikz}{}{\input{./tikz/ax/add-copy-bimonoid-1.tikz}}
\tikzset{x=1em, y=1.5ex}
 \qquad 
\tikzset{x=1em, y=2.1ex}
\InputIfFileExists{ax/add-copy-bimonoid-unit.tikz}{}{\input{./tikz/ax/add-copy-bimonoid-unit.tikz}}
\tikzset{x=1em, y=1.5ex}
 \,\;\;\,\myeq{$\circ\bullet$-biun}\;\;\;\, 
\tikzset{x=1em, y=2.1ex}
\InputIfFileExists{ax/add-bimonoid-unit-1.tikz}{}{\input{./tikz/ax/add-bimonoid-unit-1.tikz}}
\tikzset{x=1em, y=1.5ex}
 \qquad 
\tikzset{x=1em, y=2.1ex}
\InputIfFileExists{ax/add-copy-bimonoid-counit.tikz}{}{\input{./tikz/ax/add-copy-bimonoid-counit.tikz}}
\tikzset{x=1em, y=1.5ex}
 \;\;\;\myeq{$\bullet\circ$-biun}\;\;\;\, 
\tikzset{x=1em, y=2.1ex}
\InputIfFileExists{ax/add-copy-bimonoid-counit-1.tikz}{}{\input{./tikz/ax/add-copy-bimonoid-counit-1.tikz}}
\tikzset{x=1em, y=1.5ex}
\qquad
\tikzset{x=1em, y=2.1ex}
\InputIfFileExists{ax/bone-white-black.tikz}{}{\input{./tikz/ax/bone-white-black.tikz}}
\tikzset{x=1em, y=1.5ex}
\;\;\myeq{$\circ\bullet$-bo}\;\;\;
\tikzset{x=1em, y=2.1ex}
\InputIfFileExists{empty-diag.tikz}{}{\input{./tikz/empty-diag.tikz}}
\tikzset{x=1em, y=1.5ex}

\end{equation*}
\hdashrule{\linewidth}{1pt}{1pt}
\vspace{1pt}
\begin{equation*}

\tikzset{x=1em, y=2.1ex}
\InputIfFileExists{ax/copy-Frobenius-left.tikz}{}{\input{./tikz/ax/copy-Frobenius-left.tikz}}
\tikzset{x=1em, y=1.5ex}
\;\;\myeq{$\bullet$-fr1}\;\; 
\tikzset{x=1em, y=2.1ex}
\InputIfFileExists{ax/copy-Frobenius.tikz}{}{\input{./tikz/ax/copy-Frobenius.tikz}}
\tikzset{x=1em, y=1.5ex}
\;\;\myeq{$\bullet$-fr2}\;\; 
\tikzset{x=1em, y=2.1ex}
\InputIfFileExists{ax/copy-Frobenius-right.tikz}{}{\input{./tikz/ax/copy-Frobenius-right.tikz}}
\tikzset{x=1em, y=1.5ex}
 \qquad 
\tikzset{x=1em, y=2.1ex}
\InputIfFileExists{ax/copy-special.tikz}{}{\input{./tikz/ax/copy-special.tikz}}
\tikzset{x=1em, y=1.5ex}
\myeq{$\bullet$-sp}
\tikzset{x=1em, y=2.1ex}
\InputIfFileExists{ax/id.tikz}{}{\input{./tikz/ax/id.tikz}}
\tikzset{x=1em, y=1.5ex}
\qquad 
\tikzset{x=1em, y=2.1ex}
\InputIfFileExists{ax/bone-black.tikz}{}{\input{./tikz/ax/bone-black.tikz}}
\tikzset{x=1em, y=1.5ex}
\;\;\myeq{$\bullet$-bo}\;\;
\tikzset{x=1em, y=2.1ex}
\InputIfFileExists{empty-diag.tikz}{}{\input{./tikz/empty-diag.tikz}}
\tikzset{x=1em, y=1.5ex}

\end{equation*}
\vspace{3pt}
\begin{equation*}

\tikzset{x=1em, y=2.1ex}
\InputIfFileExists{ax/add-Frobenius-left.tikz}{}{\input{./tikz/ax/add-Frobenius-left.tikz}}
\tikzset{x=1em, y=1.5ex}
\;\;\myeq{$\circ$-fr1}\;\; 
\tikzset{x=1em, y=2.1ex}
\InputIfFileExists{ax/add-Frobenius.tikz}{}{\input{./tikz/ax/add-Frobenius.tikz}}
\tikzset{x=1em, y=1.5ex}
\;\;\myeq{$\circ$-fr2}\;\; 
\tikzset{x=1em, y=2.1ex}
\InputIfFileExists{ax/add-Frobenius-right.tikz}{}{\input{./tikz/ax/add-Frobenius-right.tikz}}
\tikzset{x=1em, y=1.5ex}
 
\qquad 
\tikzset{x=1em, y=2.1ex}
\InputIfFileExists{ax/add-special.tikz}{}{\input{./tikz/ax/add-special.tikz}}
\tikzset{x=1em, y=1.5ex}
\myeq{$\circ$-sp}
\tikzset{x=1em, y=2.1ex}
\InputIfFileExists{ax/id.tikz}{}{\input{./tikz/ax/id.tikz}}
\tikzset{x=1em, y=1.5ex}
\qquad
\tikzset{x=1em, y=2.1ex}
\InputIfFileExists{ax/bone-white.tikz}{}{\input{./tikz/ax/bone-white.tikz}}
\tikzset{x=1em, y=1.5ex}
\;\;\myeq{$\circ$-bo}\;\;
\tikzset{x=1em, y=2.1ex}
\InputIfFileExists{empty-diag.tikz}{}{\input{./tikz/empty-diag.tikz}}
\tikzset{x=1em, y=1.5ex}

\end{equation*}
\vspace{2pt}
\hdashrule{\linewidth}{1pt}{1pt}
\begin{align*}
  
\tikzset{x=1em, y=2.1ex}
\InputIfFileExists{ax/reals-add.tikz}{}{\input{./tikz/ax/reals-add.tikz}}
\tikzset{x=1em, y=1.5ex}
\;\myeq{add}\;
\tikzset{x=1em, y=2.1ex}
\InputIfFileExists{ax/reals-add-1.tikz}{}{\input{./tikz/ax/reals-add-1.tikz}}
\tikzset{x=1em, y=1.5ex}
 \qquad 
\tikzset{x=1em, y=2.1ex}
\InputIfFileExists{ax/zero.tikz}{}{\input{./tikz/ax/zero.tikz}}
\tikzset{x=1em, y=1.5ex}
\;\myeq{zer}\;
\tikzset{x=1em, y=2.1ex}
\InputIfFileExists{ax/reals-zero.tikz}{}{\input{./tikz/ax/reals-zero.tikz}}
\tikzset{x=1em, y=1.5ex}
\\ 
  
\tikzset{x=1em, y=2.1ex}
\InputIfFileExists{ax/reals-copy.tikz}{}{\input{./tikz/ax/reals-copy.tikz}}
\tikzset{x=1em, y=1.5ex}
\;\myeq{dup}\; 
\tikzset{x=1em, y=2.1ex}
\InputIfFileExists{ax/reals-copy-1.tikz}{}{\input{./tikz/ax/reals-copy-1.tikz}}
\tikzset{x=1em, y=1.5ex}
 \qquad 
\tikzset{x=1em, y=2.1ex}
\InputIfFileExists{ax/reals-delete.tikz}{}{\input{./tikz/ax/reals-delete.tikz}}
\tikzset{x=1em, y=1.5ex}
\;\myeq{del}\;
\tikzset{x=1em, y=2.1ex}
\InputIfFileExists{ax/delete.tikz}{}{\input{./tikz/ax/delete.tikz}}
\tikzset{x=1em, y=1.5ex}

\end{align*}
\begin{equation*}
  
\tikzset{x=1em, y=2.1ex}
\InputIfFileExists{ax/reals-multiplication.tikz}{}{\input{./tikz/ax/reals-multiplication.tikz}}
\tikzset{x=1em, y=1.5ex}
\;\myeq{$\times$}\;
\tikzset{x=1em, y=2.1ex}
\InputIfFileExists{ax/reals-multiplication-1.tikz}{}{\input{./tikz/ax/reals-multiplication-1.tikz}}
\tikzset{x=1em, y=1.5ex}
 \qquad    
\tikzset{x=1em, y=2.1ex}
\InputIfFileExists{ax/reals-sum.tikz}{}{\input{./tikz/ax/reals-sum.tikz}}
\tikzset{x=1em, y=1.5ex}
\;\myeq{$+$}\;
\tikzset{x=1em, y=2.1ex}
\InputIfFileExists{ax/reals-sum-1.tikz}{}{\input{./tikz/ax/reals-sum-1.tikz}}
\tikzset{x=1em, y=1.5ex}
\qquad   
\tikzset{x=1em, y=2.1ex}
\InputIfFileExists{ax/reals-scalar-zero.tikz}{}{\input{./tikz/ax/reals-scalar-zero.tikz}}
\tikzset{x=1em, y=1.5ex}
\;\myeq{$0$}\;
\tikzset{x=1em, y=2.1ex}
\InputIfFileExists{ax/reals-scalar-zero-1.tikz}{}{\input{./tikz/ax/reals-scalar-zero-1.tikz}}
\tikzset{x=1em, y=1.5ex}

\end{equation*}
\hdashrule{\linewidth}{1pt}{1pt}

\begin{equation*}
    
\tikzset{x=1em, y=2.1ex}
\InputIfFileExists{ax/scalar-division.tikz}{}{\input{./tikz/ax/scalar-division.tikz}}
\tikzset{x=1em, y=1.5ex}
\;\myeq{$r$-inv}\; 
\tikzset{x=1em, y=2.1ex}
\InputIfFileExists{ax/id.tikz}{}{\input{./tikz/ax/id.tikz}}
\tikzset{x=1em, y=1.5ex}
 \qquad 
\tikzset{x=1em, y=2.1ex}
\InputIfFileExists{ax/id.tikz}{}{\input{./tikz/ax/id.tikz}}
\tikzset{x=1em, y=1.5ex}
\;\myeq{$r$-coinv}\;
\tikzset{x=1em, y=2.1ex}
\InputIfFileExists{ax/scalar-co-division.tikz}{}{\input{./tikz/ax/scalar-co-division.tikz}}
\tikzset{x=1em, y=1.5ex}
\quad \text{ for } r\neq 0, r\in \field
\end{equation*}
\vspace{3pt}
\hdashrule{\linewidth}{1pt}{1pt}

\begin{equation*}

\tikzset{x=1em, y=2.1ex}
\InputIfFileExists{ax/one-copy.tikz}{}{\input{./tikz/ax/one-copy.tikz}}
\tikzset{x=1em, y=1.5ex}
\quad \myeq{1-dup}\quad 
\tikzset{x=1em, y=2.1ex}
\InputIfFileExists{ax/one2.tikz}{}{\input{./tikz/ax/one2.tikz}}
\tikzset{x=1em, y=1.5ex}
\qquad\qquad\qquad\qquad

\tikzset{x=1em, y=2.1ex}
\InputIfFileExists{ax/one-delete.tikz}{}{\input{./tikz/ax/one-delete.tikz}}
\tikzset{x=1em, y=1.5ex}
\quad \myeq{1-del}\quad 
\tikzset{x=1em, y=2.1ex}
\InputIfFileExists{empty-diag.tikz}{}{\input{./tikz/empty-diag.tikz}}
\tikzset{x=1em, y=1.5ex}

\end{equation*}
\begin{equation*}

\tikzset{x=1em, y=2.1ex}
\InputIfFileExists{ax/one-false.tikz}{}{\input{./tikz/ax/one-false.tikz}}
\tikzset{x=1em, y=1.5ex}
\quad \myeq{$\varnothing$}\quad
\tikzset{x=1em, y=2.1ex}
\InputIfFileExists{ax/one-false-disconnect.tikz}{}{\input{./tikz/ax/one-false-disconnect.tikz}}
\tikzset{x=1em, y=1.5ex}

\end{equation*}
\vspace{3pt}
\hdashrule{\linewidth}{1pt}{1pt}
\begin{equation*}

\tikzset{x=1em, y=2.1ex}
\InputIfFileExists{generators/co-one.tikz}{}{\input{./tikz/generators/co-one.tikz}}
\tikzset{x=1em, y=1.5ex}
\;\;\myeq{co1}\;\;
\tikzset{x=1em, y=2.1ex}
\InputIfFileExists{ax/co-one-def.tikz}{}{\input{./tikz/ax/co-one-def.tikz}}
\tikzset{x=1em, y=1.5ex}
\qquad\quad 
\tikzset{x=1em, y=2.1ex}
\InputIfFileExists{generators/co-scalar-r.tikz}{}{\input{./tikz/generators/co-scalar-r.tikz}}
\tikzset{x=1em, y=1.5ex}
\;\;\myeq{coreg}\;\;
\tikzset{x=1em, y=2.1ex}
\InputIfFileExists{ax/co-scalar-def.tikz}{}{\input{./tikz/ax/co-scalar-def.tikz}}
\tikzset{x=1em, y=1.5ex}

\end{equation*}
\caption{Axioms of  Affine Interacting Hopf Algebras ($\AIH{\field}$).\label{fig:ih}}
\end{figure*}

\begin{itemize}
\item In the first block, both the black and white structures are commutative monoids and comonoids, expressing fundamental properties of addition and copying.
\item In the second block, the white monoid and black comonoid interact as a bimonoid. Bimonoids are one of two canonical ways that monoids and comonoids interact, as shown in~\cite{Lack2004a}.
\item In the third and fourth block, both the black and the white monoid/comonoid pair form an extraspecial Frobenius
monoid. The Frobenius equations (\textsf{fr 1}) and (\textsf{fr 2}) are a famous algebraic pattern which establishes a bridge between algebraic and topological phenomena, see~\cite{Carboni1987,Kock2003,Coecke2017}. The ``extraspecial'' refers to the two additional equations, the \emph{special} equation ($\bullet$\textsf{-sp}) and the \emph{bone} equation ($\bullet$-\textsf{bo}). The Frobenius equations, together with the special equation, are the another canonical pattern of interaction between monoids and comonoids identified in~\cite{Lack2004a}. Together with the bone equation, the set of four equations characterises \emph{corelations}, see~\cite{Bruni01somealgebraic,Zanasi16,CF}.
\item The equations in the fourth block are parametrised over $r \in \field$ and describe commutativity of $\tscalar{r}$ with respect to the other operations, as well as multiplication and addition of scalars.  
\item The fifth block encodes multiplicative inverses of the field, guaranteeing that $\tcoscalar{r}$ behaves  as division by $r$.
\item The sixth block deals with the truly affine part of the calculus, the constant $\One$ and its relationship to other generators. The first two equations just say that $\One$ can be copied and deleted by the black structure, in other words that it denotes a single value. More interestingly, the third equation of this block is justified by the possibility of expressing the empty set, by, for example,
\begin{equation}
  \label{eq:empty-set}
  \dsem{
\tikzset{x=1em, y=2.1ex}
}
\tikzset{x=1em, y=1.5ex}
}\,=\{(\bullet,1)\}\poi\{(0,\bullet)\} = \varnothing.
\end{equation}
The last equation thus guarantees that this diagram behaves like logical false, since for any $c$ and $d$ in, $\varnothing \tns c = \varnothing \tns d = \varnothing$; composing or taking the monoidal product of $\varnothing$ with any relation results in $\varnothing$.
\item Finally, the last block constrains the mirror generators for $\One$ and $\tscalar{r}$ to be obtained from them by bending the wires around, using the black Frobenius structure. Note that there is some redundancy in the presentation, as we could have used only $\One$ and $\tscalar{r}$ as generators, and taken these to be definitions of $\coOne$ and $\tcoscalar{r}$. 
 \end{itemize}

\section{From Matrices to Circuits and Back}
\label{sec:mat}

Several proofs in Section~\ref{sec:realisability} exploit the ability to represent matrices and vectors in the graphical syntax. Details can be found in~\cite[Sec. 3.2]{ZanasiThesis} but we recall the basics below.

Roughly speaking, for any field $\field$ the theory of matrices lives inside both $\IH{\field}$ as the subprop generated by $\Bcomult,\Bcounit,\Wmult,\Wunit$ along with the scalars $\scalar$. It means that, using only these we can represent any matrix with coefficients in $\field$. And, moreover, reasoning about them can be done entirely graphically, as the corresponding equational theory is complete.

To develop some intuition for this correspondence, let us demonstrate how matrices are represented diagrammatically. Vectors can just be seen as $m\times 1$ matrices. An $m\times n$ matrix $\mathcal{M}_{\scriptscriptstyle d}$ corresponds to a diagram $d$ with $n$ wires on the left and $m$ wires on the right---the left ports can be interpreted as the columns and the right ports as the rows of $\mathcal{M}_{\scriptscriptstyle d}$. The left $j$th port is connected to the $i$th port on the right through an $r$-weighted wire whenever coefficient $({\mathcal{M}_{\scriptscriptstyle d}})_{ij}$ is a nonzero scalar $r\in \field$. When the $({\mathcal{M}_{\scriptscriptstyle d}})_{ij}$ entry is $0$, they are disconnected. Since composition along a wire carries the multiplicative structure of $\field$, we can simply draw the connection as a plain wire if $({\mathcal{M}_{\scriptscriptstyle d}})_{ij}=1$. For example,
\begin{equation*}
\text{The matrix }
  \label{eq:matrix-diagram}
  \mathcal{M}_{\scriptscriptstyle d} = 
 \begin{pmatrix}
  a & 0 & 0 \\
  b & 0 & 1 \\
  1 & 0 & 0 \\
  0 & 0 & 0 
 \end{pmatrix}
\end{equation*}
is represented by the following diagram:
\[d = 
\tikzset{x=1em, y=2.1ex}
\InputIfFileExists{ex-matrix.tikz}{}{\input{./tikz/ex-matrix.tikz}}
\tikzset{x=1em, y=1.5ex}
\]
Conversely, given a diagram, we recover the matrix by counting weighted paths from left to right ports. Then we have $\dsem{d} = \{ (v,\mathcal{M}_{\scriptscriptstyle d}\cdot v) \mid v \in \field^n\}$.

\section{Missing Proofs}\label{app:proofs}

\begin{proof}[Proof of Theorem~\ref{thm:opsem-morphism}]
We need to prove that 
\begin{enumerate}
\item $\osem{c \poi d}=\osem{c}\poi \osem{d}$,
\item $\osem{c\tns d}=\osem{c} \tns \osem{d}$.
\end{enumerate}
We prove 1. here. 2. is completely analogous.
\medskip

By rule for $\poi$ in \figref{fig:operationalSemantics},
 for each $\zeta \in \osem{c;d}$ there exist two infinite computations starting at time $t$,
$$\context{t} c_0 \dtrans{u_t}{v_t}c_1\dtrans{u_{t+1}}{v_{t+1}}c_2\dtrans{u_{t+2}}{v_{t+2}} \dots \quad\text{ and }\quad \context{t} d_0 \dtrans{v_t}{w_t}d_1\dtrans{v_{t+1}}{w_{t+1}}d_2\dtrans{v_{t+2}}{w_{t+2}}\dots $$ 
for $c_0$ and $d_0$ the initial states of $c$ and $d$ respectively, and such that $\zeta(i)=(u_i,w_i)$.
Then,  by defining $\sigma(i)=(u_i,v_i)$ and $\tau(i)=(v_i,w_i)$ for all $i\in \N$, we have that $\sigma \in \osem{c}$ and $\tau \in \osem{d}$. By construction, $\sigma$ is compatible with  $\tau$ and $\zeta=\sigma ; \tau$. Therefore $\zeta \in \osem{c};\osem{\tau}$.

Conversely, suppose that  $\zeta \in \osem{c}\poi\osem{\tau}$. Then there exists $\sigma\in \osem{c}$ and $\tau \in \osem{d}$ such that $\sigma$ and $\tau$ are compatible and $\zeta=\sigma\poi\tau$. This means that there exist
two infinite computations (starting at potentially different times),
$$\context{t} c_0 \dtrans{u_t}{v_t}c_1\dtrans{u_{t+1}}{v_{t+1}}c_2\dtrans{u_{t+2}}{v_{t+2}} \dots \quad\text{ and }\quad \context{s} d_0 \dtrans{v_t}{w_t}d_1\dtrans{v_{t+1}}{w_{t+1}}d_2\dtrans{v_{t+2}}{w_{t+2}}\dots $$ 
for $c_0$ and $d_0$ the initial states of $c$ and $d$ respectively, such that $\sigma(i)=(u_{i},v_{i})$ for $i \geq t$ and $\tau(i)=(v_{i},w_{i})$ for $i \geq s$. 

Without loss of generality, we can assume that $t\leq s$. We now have two cases: either $t < s$ or $t=s$.
\begin{itemize}
\item If $t<s$, since $s\leq 0$ by assumption (cf. Definition \ref{def:trajectories}) we can apply Lemma \ref{lemma:idle} iteratively to extend the computation 
$$\context{s} d_0 \dtrans{v_t}{w_t}d_1\dtrans{v_{t+1}}{w_{t+1}}d_2\dtrans{v_{t+2}}{w_{t+2}}\dots $$ 
by $q-p$ $\dtrans{0}{0}$ transitions into the past, to obtain 
$$\context{t} d_0 \dtrans{0}{0} d_0 \dtrans{0}{0} \dots d_0 \dtrans{v_t}{w_t}d_1\dtrans{v_{t+1}}{w_{t+1}}d_2\dtrans{v_{t+2}}{w_{t+2}}\dots $$
Clearly, the trajectory associated to this computation is still $\tau$ since $\tau(i) = 0$ for $i<s$. We have now reduced the problem to the next case.
\item If $t=s$, since $\zeta=\sigma\poi\tau$, then $\zeta(i)=(u_i,w_i)$. By the rule for $\poi$ in \figref{fig:operationalSemantics}, there exists an infinite computation
$$\context{t} e_0 \dtrans{u_t}{w_t}e_1\dtrans{u_{t+1}}{w_{t+1}}e_2\dtrans{u_{t+2}}{w_{t+2}} \dots$$
for $e_o$ the initial state of $c\poi d$. 
\end{itemize}
Therefore $\zeta\in \osem{c \poi d}$.
\end{proof}

\begin{proof}[Proof of Proposition~\ref{prop:rational-map}]
Suppose that $A$ is rational. Since rational fractions of polynomials form a ring, the multiplication and addition of two rational fractions is still rational, and therefore $A\cdot r\in\ratio^m$ for all $r\in\ratio^n$. 

Conversely, suppose that $A\cdot r\in\ratio^m$ for all $r\in\ratio^n$.
Suppose now that $A$ has an non rational coefficient, say in column $i$. Let $e_i = \begin{pmatrix}0 \dots 0 & 1 & 0 \dots 0\end{pmatrix}^T$ be the element of the canonical basis fo $\frpoly^n$ where the single $1$ is at position $i$. Then $A\cdot e_i$ returns the $i$th column of $A$, which contains an non rational coefficient by hypothesis. As a result, all coefficients of $A$ have to be rational as required.
\end{proof}

\begin{proof}[Proof of Proposition~\ref{prop:affine-rational}]
Suppose that $f\colon p\mapsto A\cdot p+b$ is rational. Since rational fractions of polynomials form a ring, the multiplication and addition of two rational fractions is still rational, and therefore $f(r)\in\ratio^m$ for all $r\in\ratio^n$. 

Conversely, suppose that $f(r)\in\ratio^m$ for all $r\in\ratio^n$. Then $f(0) = b\in\ratio^m$. We can reason as in the proof of Proposition~\ref{prop:rational-map} to prove that $A$ must be rational and conclude that $f$ is a rational affine map.
\end{proof}

\begin{proof}[Proof of Proposition~\ref{prop:single-foot}]
First to see that we can eliminate all $\coOne$, it is sufficient to notice that
\[\coOne\quad \eqAIH\quad 
\tikzset{x=1em, y=2.1ex}
\InputIfFileExists{ax/co-one-def.tikz}{}{\input{./tikz/ax/co-one-def.tikz}}
\tikzset{x=1em, y=1.5ex}
\]
Now we are left with only $\One$, we can prove the statement by induction on the number of $\One$.
If $c$ contains no $\One$, then let
\[ c' =\;\; 
\tikzset{x=1em, y=2.1ex}
\InputIfFileExists{c-after-del.tikz}{}{\input{./tikz/c-after-del.tikz}}
\tikzset{x=1em, y=1.5ex}
\]
By (1-\textsf{del}), $\dsem{c}=\dsem{c'}$.
Assume that the proposition is true for some nonnegative integer $n$. Then, given $c$ with $n+1$ $\One$, we can use the symmetric monoidal structure to pull one $\One$ through and write $c$ as follows:
\[
\tikzset{x=1em, y=2.1ex}
\InputIfFileExists{c-n-ones.tikz}{}{\input{./tikz/c-n-ones.tikz}}
\tikzset{x=1em, y=1.5ex}
\]
for some circuit $d$ with $n$ $\One$. We can apply the induction hypothesis to $d$ and get $d'$ with a single $\One$ such that $\dsem{d}=\dsem{d'}$. Thus, by the same reasoning as before, there exists $d''$ with no $\One$ such that
\[ c \eqAIH\;\; 
\tikzset{x=1em, y=2.1ex}
\InputIfFileExists{two-ones.tikz}{}{\input{./tikz/two-ones.tikz}}
\tikzset{x=1em, y=1.5ex}
\quad\myeq{1-dup}\quad
\tikzset{x=1em, y=2.1ex}
\InputIfFileExists{single-one.tikz}{}{\input{./tikz/single-one.tikz}}
\tikzset{x=1em, y=1.5ex}
\]
Since the last diagram contains only a single $\One$, we are done.
\end{proof}

\begin{proof}[Proof of Theorem~\ref{thm:affine-realisability}]
Assume that $\One$ is connected to an input port. Then, $\hat{c}$ can be rewired to a signal flow graph of which that port is an input. Then we have obtained a rewiring of $c$ as a circuit of $\ASFG$.

Conversely, assume that the circuit $c$ is realisable. Then it can be rewired to an equivalent affine signal flow graph $f$. Then, by definition of $\ASFG$, the constant $\One$ can only appear as an input of $\hat{f}$. Finally, $\hat{f}$ is a signal flow graph that is a rewiring of $\hat{c}$, so this same port is also an input of $\hat{c}$.
\end{proof}

\end{document}